\documentclass[acmtocl]{acmtrans2m}

\usepackage{epsfig,color,amssymb,amsmath,subfigure}

\newcommand{\TS}{\mathcal{T}}

\newcommand{\id}{\mathit{id}}

\newcommand{\vect}[1]{\mathbf{#1}}

\newtheorem{theorem}{Theorem}[section]

\newtheorem{proposition}[theorem]{Proposition}
\newtheorem{lemma}[theorem]{Lemma}
\newtheorem{example}[theorem]{Example}

\newdef{definition}[theorem]{Definition}
\newdef{remark}[theorem]{Remark}

\newcommand{\nats}{\mathbb{N}}

\newcommand{\infi}{\mbox{\it inf\/}}

\title{On (Omega-)Regular Model Checking}

\author{Axel Legay \\ Carnegie Mellon University \\ Computer Science
  Department \\ Pittsbugh, USA \and Pierre
  Wolper\\ Universi\'e de Li\`ege \\ Institut Montefiore, B28 \\ 4000
  Li\`ege, Belgium}




\begin{document}

\begin{bottomstuff}
Authors' e-mail~: {\tt \{legay,pw\}@montefiore.ulg.ac.be}\\
Authors' website~:
	{\tt http://www.montefiore.ulg.ac.be/}$\sim$%
	{\tt \{legay,pw\}/}\\
Axel Legay is supported by a B.A.E.F. grant.\\ 
The present article is an improved version of \cite{BLW03},
\cite{BLW04a}, and \cite{Leg08}.
\medskip
\end{bottomstuff}

\begin{abstract}
Checking infinite-state systems is frequently done by encoding
infinite sets of states as regular languages. Computing such a regular
representation of, say, the set of reachable states of a system
requires acceleration techniques that can finitely compute the effect
of an unbounded number of transitions. Among the acceleration
techniques that have been proposed, one finds both specific and
generic techniques. Specific techniques exploit the particular type of
system being analyzed, e.g. a system manipulating queues or integers,
whereas generic techniques only assume that the transition relation is
represented by a finite-state transducer, which has to be iterated.
In this paper, we investigate the possibility of using generic
techniques in cases where only specific techniques have been exploited
so far. Finding that existing generic techniques are often not
applicable in cases easily handled by specific techniques, we have
developed a new approach to iterating transducers. This new approach
builds on earlier work, but exploits a number of new conceptual and
algorithmic ideas, often induced with the help of experiments, that
give it a broad scope, as well as good performances.
\end{abstract}

\maketitle

\section{Introduction}

At the heart of all the techniques that have been proposed for
exploring infinite state spaces, is a symbolic representation that can
finitely represent infinite sets of states. In early work on the
subject, this representation was domain specific, for example linear
constraints for sets of real vectors. For several years now, the idea
that a generic finite-automaton based representation could be used in
many settings has gained ground, starting with systems manipulating
queues and integers~\cite{WB95,BGWW97,WB98,WB00}, then moving to
parametric systems~\cite{KMMPS97}, and, finally, reaching systems
using real variables~\cite{BRW98,BJW01,BJW05,BW02}.

For exploring an infinite state space, one does not only need a finite
representation of infinite sets, but also techniques for finitely
computing the effect of an unbounded number of transitions.  Such
techniques can be domain specific or generic. Domain specific
techniques exploit the specific properties and representations of the
domain being considered and were, for instance, obtained for queues
in~\cite{BG96,BH97}, for integers and reals
in~\cite{Boi98,BW02,BHJ03,BH06,FL02,BFL04,BFLS05}, for pushdown system
in~\cite{FWW97,BEM97}, and for lossy channels in~\cite{AJ96}. Generic
techniques consider finite-automata representations and provide
algorithms that operate directly on this representation, mostly
disregarding the domain for which it is used.

Generic techniques appeared first in the context of the verification
of systems whose states can be encoded by {\em finite words}, such as
parametric systems. The idea used there is that a configuration being
a finite word, a transition relation is a relation on finite words, or
equivalently a language of pairs of finite words. If this language is
regular, it can be represented by a finite state automaton, more
specifically a finite-state {\em transducer}, and the problem then
becomes the one of iterating such a transducer. Finite state
transducers are quite powerful (the transition relation of a Turing
machine can be modeled by a finite-state transducer), the flip side of
the coin being that the iteration of such a transducer is neither
always computable, nor regular. Nevertheless, there are a number of
practically relevant cases in which the iteration of finite-state
transducers can be computed and remains finite-state. Identifying such
cases and developing (partial) algorithms for iterating finite-state
transducers has been the topic, referred to as ``Regular Model
Checking'', of a series of recent
papers~\cite{KMMPS97,BJNT00,BLW03,BLW04a,JN00,BHV04,Tou01,DLS02,AJNd03}.

The question that initiated the work presented in this paper is,
whether the generic techniques for iterating transducers could be
fruitfully applied in cases in which domain specific techniques had
been exclusively used so far. In particular, one of our goals was to
iterate finite-state transducers representing arithmetic relations
(see~\cite{BW02} for a survey). Beyond mere curiosity, the motivation
was to be able to iterate relations that are not in the form required
by the domain specific results, for instance disjunctive
relations. Initial results were very disappointing: the transducer for
an arithmetic relation as simple as $(x,x+1)$ could not be iterated by
existing generic techniques. However, looking for the roots of this
impossibility through a mix of experiments and theoretical work, and
taking a pragmatic approach to solving the problems discovered, we
were able to develop an approach to iterating transducers that easily
handles arithmetic relations, as well as many other
cases. Interestingly, it is by using a tool for manipulating automata
(LASH~\cite{LASH}), looking at examples beyond the reach of manual
simulation, and testing various algorithms that the right intuitions,
later to be validated by theoretical arguments, were developed.

The general approach that has been taken is similar to the one
of~\cite{Tou01} in the sense that, starting with a transducer $T$, we
compute powers $T^i$ of $T$ and attempt to generalize the sequence of
transducers obtained in order to capture its infinite union. This is
done by comparing successive powers of $T$ and attempting to
characterize the difference between powers of $T$ as a set of states
and transitions that are added. If this set of added states, or {\em
  increment}, is always the same, it can be inserted into a loop in
order to capture all powers of $T$. However, for arithmetic
transducers comparing $T^i$ with $T^{i+1}$ did not yield an increment
that could be repeated, though comparing $T^{2^i}$ with $T^{2^{i+1}}$
did. So, a first idea we used is not to always compare $T^i$ and
$T^{i+1}$, but to extract a sequence of samples from the sequence of
powers of the transducer, and work with this sequence of
samples. Given the binary encoding used for representing arithmetic
relations, sampling at powers of $2$ works well in this case, but the
sampling approach is general and different sample sequences can be
used in other cases. Now, if we only consider sample powers $T^{i_k}$
of the transducers and compute $\bigcup_k T^{i_k}$, this is not
necessarily equivalent to computing $\bigcup_i T^{i}$. Fortunately,
this problem is easily solved by considering the reflexive transducer,
i.e., $T_0 = T\cup T_{Id}$ where $T_{Id}$ is the identity transducer,
in which case working with an infinite subsequence of samples is
sufficient.

Once the automata in the sequence being considered are constructed and
compared, and that an increment corresponding to the difference
between successive elements has been identified, the next step is to
allow this increment to be repeated an arbitrary number of times by
incorporating it into a loop. There are some technical issues about
how to do this, but no major difficulty. Once the resulting
``extrapolated'' transducer has been obtained, one still needs to
check that the applied extrapolation is safe (contains all elements of
the sequence) and is precise (contains no more).  An easy to check
sufficient condition for the extrapolation to be safe is that it
remains unchanged when being composed with itself. Checking
preciseness is more delicate, but we have developed a procedure that
embodies a sufficient criterion for doing so. The idea is to check
that any behavior of the transducer with a given number $k$ of copies
of the increment, can be obtained by composing transducers with less
than $k$ copies of the increment. This is done by augmenting the
transducers to be checked with counters and proving that one can
restrict theses counters to a finite range, hence allowing
finite-state techniques to be used.

Taking advantage of the fact that our extrapolation technique works on
automata, not just on transducers, we consider computing reachable
states both by computing the closure of the transducer representing
the transition relation, and by repeatedly applying the transducer to
a set of initial states. The first approach yields a more general
object and is essential if one wishes to extend the method to the
verification of temporal properties
(\cite{BJNT00,PS00,AJNdS04,BLW04b}), but the second is often less
demanding from a computational point of view and can handle cases that
are out of reach for the first. Preciseness is not always possible to
check when working with state sets rather than transducers, but this
just amounts to saying that what is computed is possibly an
overapproximation of the set of reachable states, a situation which is
known to be pragmatically unproblematic.

Going further, the problem of using Regular Model Checking technique
for systems whose states are represented by infinite (omega) words has
been addressed. This makes the representation of sets of reals
possible as described in~\cite{BJW01,BHJ03}. To avoid the hard to
implement algorithms needed for some operations on infinite-word
automata, only omega-regular sets that can be defined by weak
deterministic B\"uchi automata~\cite{MSS86} are considered. This is of
course restrictive, but as is shown in~\cite{BJW01,BJW05}, it is
sufficient to handle sets of reals defined in the first-order theory
of linear constraints. Moreover using such a representation leads to
algorithms that are very similar to the ones used in the finite word
case, and allows us to work with reduced deterministic automata as a
normal form. Due to these advantages and properties, one can show that
the technique developed for the finite word case can directly be
adapted to weak deterministic B\"uchi automata up to algorithmic
modifications.

Our technique has been implemented in a tool called T(0)RMC (Tool for
(Omega-)Regular Model Checking), which has been tested on several
classes of infinite-state systems. It is worth mentioning that the
ability of T(O)RMC to extrapolate a sequence of automata has other
applications than solving the ($\omega$-)Regular Reachability
Problems. As an example, the tool has been used in a semi-algorithm to
compute the convex hull of a set of integer vectors
\cite{CLW07,CLW08}. T(O)RMC was also used to compute a symbolic
representation of the simulation relation between the states of
several classes of infinite-state systems with the aim of verifying
temporal properties\,\cite{BLW04b}.\\
\newline
{\bf Structure of the paper.}  The paper is structured as follows. In
Section \ref{section-background}, we recall the elementary definitions
on automata theory that will be used throughout the rest of the
paper. Section \ref{counter-section} introduces {\em counter-word
  automata}, a class of counter automata that will be used by our
preciseness technique. Section \ref{paper-model} presents the
($\omega$-)Regular Model Checking framework as well as the problems we
want to solve. Sections \ref{section-solving}, \ref{section-sampling},
\ref{section-increment}, \ref{section-extrapolation}, and
\ref{section-precision} describe our main results. Implementation and
experiments are discussed in Section
\ref{section-implementation}. Finally, Sections
\ref{section-comparison} and \ref{section-conclu} contain a comparison
with other works on the same topic and several directions for future
research, respectively.

\section{Background on Automataa Theory}
\label{section-background}
In this section, we introduce several notations, concepts, and
definitions that will be used throughout the rest of this paper. The
set of natural numbers is denoted by $\nats$, and $\nats_0$ is used
for $\nats\setminus {\lbrace}0{\rbrace}$.

\subsection{Relations}

Consider a set $S$, a set $S_1\,\subseteq\,S$, and two
binary{\footnote{The term ``binary'' will be dropped in the rest of
    the paper.}} relations $R_1, R_2\,\subseteq\,S\times S$. The
identity relation on $S$, denoted $R_\id^S$ (or $R_{id}$ when $S$ is
clear from the context) is the set ${\lbrace}(s,s)|s\in
S{\rbrace}$. The {\em image} of $S_1$ by $R_1$, denoted $R_1(S_1)$, is
the set ${\lbrace}s'\in S_1\mid (\exists s\in S_1)( (s,s')\in
R_1){\rbrace}$. The {\em composition} of $R_1$ with $R_2$, denoted
$R_2\circ R_1$, is the set ${\lbrace}(s,s')\mid (\exists
s'')((s,s'')\in R_1\wedge (s'',s')\in R_2){\rbrace}$. The $i$th {\em
  power} of $R_1$ ($i\in \nats_0$), denoted $R_1^i$, is the relation
obtained by composing $R_1$ with itself $i$ times. The {\em
  zero-power} of $R_1$, denoted $R_1^0$, corresponds to the identity
relation. The {\em transitive closure} of $R_1$, denoted $R_1^+$, is
given by $\bigcup_{i=1}^{i=+\infty}R_1^i$, its {\em reflexive
  transitive closure}, denoted $R^*$, is given by $R_1^+\cup
R_{\id}^S$. The {\em domain} of $R_1$, denoted $\it{Dom}(R_1)$, is
given by ${\lbrace}s\in S\mid (\exists s'\in S)((s,s')\in
R_1){\rbrace}$.

\subsection{Words and Languages}
\label{word-language}

An {\em alphabet} is a (nonempty) finite set of distinct symbols.  A
{\em finite word\/} of length $n$ over an alphabet $\Sigma$ is a
mapping $w:{\lbrace}0,{\dots},n-1{\rbrace}{\rightarrow}\Sigma$. An
{\em infinite word \/}, also called $\omega-$word, over $\Sigma$ is a
mapping $w:{\nats}{\rightarrow}\Sigma$. We denote by the term {\em
  word} either a finite word or an infinite word, depending on the
context. The {\em length} of the finite word $w$ is denoted by
$|w|$. A finite word $w$ of length $n$ is often represented by
$w=w(0){\cdots}w(n-1)$. An infinite word $w$ is often represented by
$w(0)w(1){\cdots}$ . The sets of finite and infinite words over
$\Sigma$ are denoted by $\Sigma^*$ and by $\Sigma^{\omega}$,
respectively. We define $\Sigma^{\infty}=\Sigma^*\cup
\Sigma^{\omega}$. A {\em finite-word (respectively infinite-word)
  language} over $\Sigma$ is a (possibly infinite) set of finite
(respectively, infinite) words over $\Sigma$. Consider $L_1$ and
$L_2$, two finite-word (resp. infinite-word) languages. The {\em
  union} of $L_1$ and $L_2$, denoted $L_1\cup L_2$, is the language
that contains all the words that belong either to $L_1$ or to
$L_2$. The {\em intersection} of $L_1$ and $L_2$, denoted $L_1\cap
L_2$, is the language that contains all the words that belong to both
$L_1$ and $L_2$. The {\em complement} of $L_1$, denoted
$\overline{L_1}$ is the language that contains all the words over
$\Sigma$ that do not belong to $L_1$.\\
\newline
We alos introduce {\em synchronous product} and {\em projection},
which are two operations needed to define relations between languages.

\begin{definition}
\label{definition-cartesian-product}
Consider $L_1$ and $L_2$ two languages over $\Sigma$.
\begin{itemize}
\item
If $L_1$ and $L_2$ are finite-word languages, the synchronous product
$L_1\bar{{\times}} L_2$ of $L_1$ and $L_2$ is defined as follows
\begin{center}
$L_1\bar{{\times}}
L_2={\lbrace}(w(0),w(0)'){\dots}(w(n),w(n)')\mid$\\$
w=w(0)w(1){\dots}w(n)\in L_1\,{\wedge}\,w'=w(0)'w(1)'{\dots}w(n)'\in
L_2{\rbrace}$.
\end{center} 
\item
If $L_1$ and $L_2$ are $\omega$-languages, the synchronous product
$L_1\bar{{\times}} L_2$ of $L_1$ and $L_2$ is defined as follows
\begin{center}
$L_1\bar{{\times}}
L_2={\lbrace}(w(0),w(0)')(w(1),w(1)'){\cdots}\mid$\\$
w=w(0)w(1){\dots}\in L_1\,{\wedge}\,w'=w(0)'w(1)'{\cdots}\in
L_2{\rbrace}$.
\end{center} 
\end{itemize}
The language $L_1\bar{{\times}} L_2$ is defined over the alphabet
$\Sigma^2$.
\end{definition}

\noindent
Definition \ref{definition-cartesian-product} directly generalizes to
synchronous products of more than two languages. Given two finite
(respectively, infinite) words $w_1, w_2$ (with $|w_1|=|w_2|$ if the
words are finite) and two languages $L_1$ and $L_2$ with
$L_1={\lbrace}w_1{\rbrace}$ and $L_2={\lbrace}w_2{\rbrace}$, we use
$w_1{\bar{\times}}w_2$ to denote the {\em unique} word in
$L_1{\bar{\times}}L_2$.

\begin{definition}
\label{def-projection}
Suppose $L$ a language over the alphabet $\Sigma^n$ and a natural
$1\,{\leq}\,i\,{\leq}n$. The projection of $L$ on all its components
except component $i$, denoted $\Pi_{\not= i}(L)$, is the language $L'$
such that
\begin{center}
$\Pi_{\not=i }(L)={\lbrace}w_1\bar{\times}\dots\bar{\times}w_{i-1}\bar{\times}w_{i+1}\bar{\times}\dots
\bar{\times} w_n \mid$\\$ (\exists
w_i)(w_1\bar{\times}\dots\bar{\times}w_{i-1}\bar{\times}w_i\bar{\times}w_{i+1}\bar{\times}\dots
\bar{\times} w_n\in L){\rbrace}$.
\end{center}
\end{definition}

\subsection{Automata}

\begin{definition}
\label{finite-word-automaton}
An automaton over $\Sigma$ is a tuple $A=(Q,\Sigma,Q_0,\triangle,F)$,
where
\begin{itemize}
\item
$Q$ is a finite set of {\em states\/},
\item
$\Sigma$ is a {\em finite} alphabet,
\item
$Q_0\,\subseteq\,Q$ is the set of {\em initial states\/},
\item
$\triangle\,\subseteq\, Q\times \Sigma \times Q$ is a finite {\em transition
relation\/}, and
\item
$F\,\subseteq\,Q$ is the set of accepting states (the states in
$Q\setminus F$ are the {\em nonaccepting} states).
\end{itemize}
\end{definition}

Let $A=(Q,\Sigma,Q_0,\triangle,F)$ be an automaton. If
$(q_1,a,q_2)\in \triangle$, then we say that there is a {\em
  transition} from $q_1$ (the {\em origin}) to $q_2$ (the {\em
  destination}) labeled by $a$. We sometimes abuse the notations, and
write $q_2\in \triangle(q_1,a)$ instead of $(q_1,a,q_2)\in
\triangle$. Two transitions $(q_1,a,q_2), (q_3,b,q_4)\in \triangle$
are {\em consecutive} if $q_2=q_3$. Given two states $q,q'\in Q$ and a
finite word $w\in \Sigma^*$, we write $(q,w,q')\in \triangle^*$ if
there exist states $q_0,\dots,q_{k-1}$ and $w_0,\dots,w_{k-2}\in
\Sigma$ such that $q_0=q$, $q_{k-1}=q'$, $w=w_0w_1\cdots w_{k-2}$, and
$(q_i,w_i,q_{i+1})\in \triangle$ for all $0\,{\leq}\,i<k-1$. Given two
states $q, q'\in Q$, we say that the state $q'$ is {\em reachable}
from $q$ in $A$ if $(q,a,q')\in \triangle^*$. The automaton $A$ is
{\em complete} if for each state $q\in Q$ and symbol $a\in \Sigma$,
there exists at least one state $q'\in Q$ such that $(q,a,q')\in
\triangle$. An automaton can easily be completed by adding an
extra nonaccepting state.\\
\newline
A {\em finite run} of $A$ on a finite word
$w:{\lbrace}0,{\dots},n-1{\rbrace}{\rightarrow}\Sigma$ is a labeling
$\rho:{\lbrace}0,{\dots},n{\rbrace}{\rightarrow}Q$ such that
$\rho(0)\in Q_0$, and
$(\forall{0\,{\leq}\,i\,{\leq}\,n-1})((\rho(i),w(i),\rho(i+1))\in
\triangle)$. A finite run $\rho$ is {\em accepting} for $w$ if
$\rho(n)\in F$. An {\em infinite run} of $A$ on an infinite word
$w:{\nats}{\rightarrow}\Sigma$ is a labeling
$\rho:\nats{\rightarrow}Q$ such that $\rho(0)\in Q_0$, and
$(\forall{0\,{\leq}\,i})((\rho(i),w(i),\rho(i+1))\in \triangle)$. An
infinite run $\rho$ is {\em accepting} for $w$ if $\infi(\rho)\cap F
\neq \emptyset$, where $\infi(\rho)$ is the set of states that are
visited infinitely often by $\rho$.\\
\newline
We distinguish between {\em finite-word automata} that are automata
accepting finite words, and {\em B\"uchi automata} that are automata
accepting infinite words. A finite-word automaton accepts a finite
word $w$ if there exists an accepting finite run for $w$ in this
automaton. A B\"uchi automaton accepts an infinite word $w$ if there
exists an accepting infinite run for $w$ in this automaton. The set of
words accepted by $A$ is the {\em language accepted by $A$}, and is
denoted $L(A)$. Any language that can be represented by a finite-word
(respectively, B\"uchi) automaton is said to be {\em regular}
(respectively, {\em $\omega$-regular}).\\
\newline
The automaton $A$ may behave nondeterministicaly on an input word,
since it may have many initial states and the transition relation may
specify many possible transitions for each state and symbol. If
$|Q_0|=1$ and for all state $q_1\in Q$ and symbol $a\in \Sigma$ there
is at most one state $q_2\in Q$ such that $(q_1,a,q_2)\in \triangle$,
then $A$ is {\em deterministic}. In order to emphasize this property,
a deterministic automaton is denoted as a tuple $(Q, \Sigma,
q_0,\delta,F)$, where $q_0$ is the unique initial state and
$\delta:Q\times \Sigma~{\rightarrow}~Q$ is a partial function deduced
from the transition relation by setting $\delta(q_1,a)=q_2$ if
$(q_1,a,q_2)\in \triangle$. Operations on languages directly translate
to operations on automata, and so do the notations.\\
\newline
One can decide weither the language accepted by a finite-word or a
B\"uchi automaton is empty or not. It is also known that finite-word
automata are closed under determinization, complementation, union,
projection, and intersection\,\cite{Hop71}. Moreover, finite-word
automata admit a minimal form, which is unique up to
isomorphism\,\cite{Hop71}.\\
\newline
Though the union, intersection, synchronous product, and projection of
B\"uchi automata can be computed efficiently, the complementation
operation requires intricate algorithms that not only are worst-case
exponential, but are also hard to implement and optimize (see
\cite{Var07a} for a survey). The core problem is that there are
B\"uchi automata that do not admit a deterministic/minimal form. To
working with infinite-word automata that do own the same properties as
finite-word automata, we will restrict ourselves to {\em weak\/}
automata~\cite{MSS86} defined hereafter.

\begin{definition}
For a B\"uchi automaton $A=(\Sigma,Q,q_0,\delta,F)$ to be weak, there
has to be partition of its state set $Q$ into disjoint subsets $Q_1,
\ldots, Q_m$ such that for each of the $Q_i$, either $Q_i \subseteq
F$, or $Q_i \cap F = \emptyset$, and there is a partial order $\leq$
on the sets $Q_1, \ldots, Q_m$ such that for every $q \in Q_i$ and
$q'\in Q_j$ for which, for some $a \in \Sigma$, $q'\in \delta(q,a)$
($q'=\delta(q,a)$ in the deterministic case), $Q_j \leq Q_i$.
\end{definition}

\noindent
A weak automaton is thus a B\"uchi automaton such that each of the
strongly connected components of its graph contains either only
accepting or only non-accepting states.\\
\newline
Not all $\omega$-regular languages can be accepted by deterministic
weak B\"uchi automata, nor even by nondeterministic weak
automata. However, there are algorithmic advantages to working with
weak automata : deterministic weak automata can be complemented simply
by inverting their accepting and non-accepting states; and there
exists a simple determinization procedure for weak
automata~\cite{Saf92}, which produces B\"uchi automata that are
deterministic, but generally not weak. Nevertheless, if the
represented language can be accepted by a deterministic weak
automaton, the result of the determinization procedure will be {\em
  inherently weak\/} according to the definition below~\cite{BJW01}
and thus easily transformed into a weak automaton.

\begin{definition}
A B\"uchi automaton is {\em inherently weak\/} if none of the
reachable strongly connected components of its transition graph
contain both accepting (visiting at least one accepting state) and
non-accepting (not visiting any accepting state) cycles.
\end{definition}

This gives us a pragmatic way of staying within the realm of
deterministic weak B\"uchi automata. We start with sets represented by
such automata. This is preserved by union, intersection, synchronous
product, and complementation operations. If a projection is needed,
the result is determinized by the known simple procedure. Then, either
the result is inherently weak and we can proceed, or it is not and we
are forced to use the classical algorithms for B\"uchi automata. The
latter cases might never occur, for instance if we are working with
automata representing sets of reals definable in the first-order
theory of linear constraints~\cite{BJW01}.\\
\newline
A final advantage of weak deterministic B\"uchi automata is that they
admit a minimal form, which is unique up to isomorphism~\cite{Loe01}.

\subsection{Relations on Automata States}

We will also use the following definitions.

\begin{definition}
\label{equivalence-relations}
Given two automata $A_1=(Q_1,\Sigma_1,Q_{01},\triangle_1,F_1)$ and
$A_2=(Q_2,$\\$\Sigma_2,Q_{02},\triangle_2,F_2)$, we define
\begin{itemize}
\item
the forward equivalence relation $E_f,\subseteq\,Q_1\times Q_2$,
which is an equivalence relation on states of $A_1$ and $A_2$ with
$(q_1,q_2)\,\in\,E_f$ iff $L_{q_1}^{F_1}(A_1)=L_{q_2}^{F_2}(A_2)$;
\item
the backward equivalence relation $E_b\,\subseteq\,Q_1\times Q_2$,
which is an equivalence relation on states of $A$ with
$(q_1,q_2)\,\in\,E_b$ iff
$L^{q_1}_{Q_{01}}(A_1)=L^{q_2}_{Q_{02}}(A_2)$.
\end{itemize}
\end{definition}

\begin{definition}
\label{iso-def}
Given two automata $A_1=(Q_1,\Sigma,Q_{01},\triangle_1,F_1)$ and
$A_2=(Q_2,\Sigma,$\\$Q_{02},\triangle_2,F_2)$, a relation
$R\,\subseteq\,Q_1\times Q_2$ is an isomorphism between $A_1$ and
$A_2$ if and only if
\begin{itemize}
\item
$R$ is a bijection,
\item
for each $a\in (\Sigma\cup {\lbrace}\epsilon{\rbrace})$ and
$q_1,q_2\in Q_1$, $(q_1,a,q_2)\in
\triangle_1\,\Leftrightarrow\,(R(q_1),a,R(q_2))\in \triangle_2$,
\item
for each $(q,q')\in R$, $q\in Q_{01}\,\Leftrightarrow\, q'\in Q_{02}$,
\item
for each $(q,q')\in R$, $q\in F_1\,\Leftrightarrow\, q'\in F_2$.
\end{itemize} 
\end{definition}

\subsection{Transducers}

In this paper, we will consider relations that are defined over sets
of words. We use the following definitions taken from
\cite{Nil01}. For a finite-word (respectively, infinite-word) language
$L$ over $\Sigma^n$, we denote by ${\lfloor}L{\rfloor}$ the
finite-word (respectively, infinite-word) relation over $\Sigma^n$
consisting of the set of tuples $(w_1,w_2,{\dots},w_n)$ such that
$w_1\bar{{\times}}w_2\bar{{\times}}\dots\bar{{\times}}w_n$ is in
$L$. The arity of such a relation is $n$. Note that for $n=1$, we have
that $L={\lfloor}L{\rfloor}$. The relation $R_{id}$ is the {\em
  identity relation}, i.e.,
$R_{id}={\lbrace}(w_1,w_2,{\dots},w_n)|w_1=w_2={\dots}=w_n{\rbrace}$. A
relation $R$ defined over $\Sigma^n$ is \\{\em ($\omega$-)regular} if
there exists a ($\omega$-)regular language $L$ over $\Sigma^n$ such
that ${\lfloor}L{\rfloor}=R$.\\
\newline
We now introduce transducers that are automata for
representing ($\omega$-)regular relations over $\Sigma^2$.

\begin{definition}
\label{definition-transdu}
A transducer over $\Sigma^2$ is an automaton $T$ over $\Sigma^2$ given
by $(Q,\Sigma^2,$\\$ Q_0,\triangle, F)$, where
\begin{itemize}
\item
$Q$ is the {\em finite} set of {\em states\/},
\item
$\Sigma^2$ is the {\em finite} alphabet,
\item
$Q_0\,\subseteq\,Q$ is the set of {\em initial states\/},
\item
$\triangle: Q\times \Sigma^2\times
Q$ is the {\em transition relation\/}, and
\item
$F\,\subseteq\,Q$ is the set of accepting states (the states that are not
in $F$ are the {\em nonaccepting} states).
\end{itemize}
\end{definition}

\noindent
Given an alphabet $\Sigma$, the transducer representing the identity
relation over $\Sigma^2$ is denoted $T_{id}^{\Sigma}$ (or $T_{id}$
when $\Sigma$ is clear from the context). All the concepts and
operations defined for finite automata can be used with
transducers. The only reason to particularize this class of automata
is that some operations, such as composition, are specific to
relations. In the sequel, we use the term ``transducer'' instead of
``automaton'' when using the automaton as a representation of a
relation rather than as a representation of a language. We sometimes
abuse the notations and write $(w_1,w_2)\in T$ instead of
$(w_1,w_2)\in {\lfloor}L(T){\rfloor}$. Given a pair $(w_1,w_2)\in T$,
$w_1$ is the {\em input word}, and $w_2$ is the {\em output word}. The
transducers we consider here are often called {\em
  structure-preserving}. Indeed, when following a transition, a symbol
of the input word is replaced by exactly one symbol of the output
word.

\begin{example}
\label{ex-xx1}
If positive integers are encoded in binary with an arbitrary number of
leading $0$'s allowed, and negative numbers are represented using
$2$'s complement allowing for an arbitrary number of leading $1$'s,
the transducer of Figure~\ref{fig-xx1a} represents the relation
$(x,x+1) \cup(x,x)$ (see~\cite{BW02} for a full description of the
encoding).

\begin{figure}
\begin{center}
\includegraphics{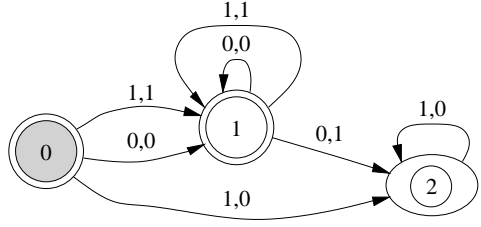}
\end{center}
\caption{A transducer for $(x,x+1)\cup (x,x)$. The initial state of
  the automaton is colored in gray, and the final state is surrounded
  by a double circle (this convention will be followed throughout the
  rest of the paper).}
\label{fig-xx1a}
\end{figure}
\end{example}

Given two transducers $T_1$ and $T_2$ over the alphabet $\Sigma$ that
represents two relations $R_1$ and $R_2$, respectively. The {\em
  composition} of $T_1$ by $T_2$, denoted $T_2\circ T_1$ is the
transducer that represents the relation $R_2\circ R_1$. We denote by
$T_1^i$ ($i\in \nats_0$) the transducer that represents the relation
$R_1^i$. The {\em transitive closure} of $T$ is
$T^+=\bigcup_{i=1}^{\infty}T^i$; its {\em reflexive transitive
  closure} is $T^*=T^+\cup T_{id}$. The transducer $T$ is {\em
  reflexive} if and only if $L(T_{id})\,\subseteq\,L(T)$. Given an
automaton $A$ over $\Sigma$ that represents a set $S$, we denote by
$T(A)$ the automaton representing the {\em image} of $A$ by $T$, i.e.,
an automaton for the set $R(S)$.\\
\newline
Let $T_1$ and $T_2$ be two finite-word (respectively, B\"uchi)
transducers defined over $\Sigma^2$ and let $A$ be a finite-word
automaton (respectively, B\"uchi) automaton defined over $\Sigma$. We
observe that $T_2\circ T_1 = \pi_{\not= 2}{\lbrack}(T_1 \bar{{\times}}
T_{id}^\Sigma)\cap (T_{id}^\Sigma \bar{{\times}} T_2){\rbrack}$ and
$T(A)=\pi_{\not= 1}{\lbrack}(A^{\Sigma} \bar{{\times}} \Sigma)\cap
T{\rbrack}$, where $A^{\Sigma}$ is an automaton accepting $\Sigma^*$
(respectively, $\Sigma^{\omega}$). As a consequence, the composition
of two finite-word ((weak) B\"uchi) transducers is a finite-word
transducer. However, the composition of two deterministic weak B\"uchi
transducer is a weak B\"uchi transducer whose deterministic version
may not be weak. A same observation can be made about the composition
of a transducer with an automaton.

\section{Counter Automata}
\label{counter-section}

We introduce {\em counter-word automata}, a class of automata whose
states are augmented by a vector of counters. Counter-word automata
are intended to be used in our procedure for checking the preciseness
of an extrapolation. All the concepts presented in this section are
thus developped for this purpose.

\subsection{Definitions}

We start with the definition of a counter automaton.

\begin{definition}
\label{counter-automaton}
A counter-word automaton (counter automaton for short) over an
alphabet $\Sigma$ is a tuple
$A_c=(n,\vect{c},Q,\Sigma,Q_0,\triangle,F)$, where
\begin{itemize}
\itemsep0cm
\item
$n\in \nats$ is the counter {\em dimension} of $A$,
\item
$\vect{c}=(c_1,{\dots},c_{n})$ is a {\em vector of counters whose
  values range over the natural numbers}.  A {\em counter valuation}
  $\vect{v}\in \nats^n$ for $\vect{c}$ is a vector of natural numbers,
  where the $ith$ component of $\vect{v}$ assigns a value to $c_i$,
\item
$Q$ is a set of {\em states\/} (unless stated otherwise, $Q$ is
  assumed to be finite),
\item
$\Sigma$ is a {\em finite} alphabet,
\item
$Q_0\,\subseteq\,Q$ is a set of {\em initial states\/},
\item
$\triangle\,\subseteq\, Q\times (\Sigma \times \nats^n)\times Q$ is a
finite {\em transition relation\/}, and
\item
$F\,\subseteq\,Q$ is a set of accepting states.
\end{itemize}
\end{definition}

\noindent
Let $A_c=(n,\vect{c},Q,\Sigma,Q_0,\triangle,F)$ be a counter
automaton. If $(q_1,(a,\vect{v}),q_2)\in \triangle$, then we say that
there is a {\em transition} from $q_1$ (the {\em origin}) to $q_2$
(the {\em destination}) labeled by $a$, and associated to the {\em
counter valuation} $\vect{v}$. The initial value of each counter is
$0$, and each time a transition is followed, the current values of the
counters are incremented with the counter valuation associated to the
transition. Given a counter automaton
$A_c=(n,\vect{c},Q,\Sigma,Q_0,\triangle,F)$, the {\em maximal
increment value} of $A_c$ is the smallest $d\in \nats$ such that
$\triangle\,\subseteq\, Q\times (\Sigma \times
{\lbrack}0,d{\rbrack}^n)\times Q$. Counter automata being finite
structures, the maximal increment value can always be computed by
enumerating the elements of the transition relation.  As finite
automata, counter automata are graphically represented with
edge-labeled directed graphs. We emphasize the counter increment
vector associated to each transition by preceding it with the symbol
``+''.

Our aim is to associate counter valuations to the words accepted by a
counter automaton. For doing so, we first define a notion of accepted
language that does not take the counters into account. We propose the
following definition.

\begin{definition}
Let $A_c=(n,\vect{c},Q,\Sigma,Q_0,\triangle,F)$ be a counter
automaton. The {\em counterless} automaton corresponding to $A_c$ is
the finite automaton $A=(Q,\Sigma,Q_0,$\\$\triangle',F)$, where
\begin{center}
$\triangle'={\lbrace}(q,a,q')\in Q'\times \Sigma\times Q'\mid
(\exists \vect{v}\in \nats^n)((q,(a,\vect{v}),q')\in
\triangle){\rbrace}$.
\end{center}
\end{definition}

\begin{definition}
The {\em language} accepted by a counter automaton $A_c$, denoted
$L(A_c)$, is the language accepted by its corresponding counterless
automaton. If $w\in L(A_c)$, then we say that $w$ is {\em accepted} by
$A_c$.
\end{definition}

We now describe how and when a counter automaton can assign counter
values to the words it accepts. Let
$A_c=(n,\vect{c},Q,\Sigma,Q_0,\triangle,F)$ be a counter
automaton. Assume first that $A_c$ describes a set of finite words. A
{\em run} of $A_c$ on a finite word
$w:{\lbrace}0,{\dots},m-1{\rbrace}{\rightarrow}\Sigma$ is a labeling
$\rho:{\lbrace}0,\dots, m{\rbrace}\rightarrow (Q\times \nats^n)$ such
that
\begin{enumerate}
\itemsep0cm
\item
$\rho(0)\in (Q_0\times \vect{0})$, and
\item
$(\forall 0\,\leq\,i\,\leq\,m-1)$,
$\rho(i+1)=(q_{i+1},\vect{v_{i+1}})$ if and only if
$\rho(i)=(q_i,\vect{v_i})$ and there exists
$(q_{i},(w(i),\vect{v}),q_{i+1})\in \triangle$ with
$\vect{v_{i+1}}=\vect{v_i}+\vect{v}$.
\end{enumerate}

\noindent
Let $\rho(m-1) = (q_f\times {\lbrace}\vect{v}{\rbrace})$. If $q_f\in
F$, then we say that $\rho$ is an {\em accepting run} and that $w$ is
{\em accepted} by $A_c$ with the counter valuation
$\vect{v}$. Otherwise $\rho$ is {\em rejecting} for $w$. The automaton
$A_c$ being a finite-word automaton, we can always associate at least
one counter valuation to each word $w\in L(A_c)$. Observe that if the
counterless automaton of $A_c$ behaves non deterministically on $w$,
then this word may be associated to several counter valuations. There
can be accepting and nonaccepting runs that assign the same counter
valuation to $w$.



We now switch to the case of infinite words. A run of $A_c$ on an
infinite word $w:\nats{\rightarrow}\Sigma$ is a labeling
$\rho:\nats\rightarrow (Q\times
\nats^n)$ such that

\begin{enumerate}
\itemsep0cm
\item
$\rho(0)\in (Q_0\times \vect{0})$, and
\item
$(\forall 0\,\leq\,i)$, $\rho(i+1)=(q_{i+1},\vect{v_{i+1}})$
if and only if $\rho(i)=(q_i,\vect{v_i})$ and there exists
$(q_{i},(w(i),\vect{v}),q_{i+1})\in
\triangle$ with $\vect{v_{i+1}}=\vect{v_i}+\vect{v}$.
\end{enumerate}

Contrary to the finite-word case, it is generally not possible to
associate a counter valuation to $\rho$. Indeed, there could be the
case that the counters are incremented an unbounded number of
times. There are however sub-classes of infinite-word counter automata
for which it is always possible to assign a counter valuation to each
of its runs. This is illustrated with the following definition.

\begin{definition}
Let $A_c=(n,Q,\Sigma,Q_{0},\triangle,F)$ be a weak B\"uchi counter
automaton. We say that $A_c$ is {\em run-bounded} if for each of
its accepting strongly connected components $S\,\subseteq\,F$ and
states $q_1,q_2\in S$, any transition that goes from $q_1$ to $q_2$ is
associated with the counter valuation $\vect{0}$.
\end{definition}

\noindent
The structure of a run-bounded weak B\"uchi counter automaton ensures
that for each of its runs, after having followed a finite number of
transitions, the values of the counters are no longer
incremented. Hence, one can reason on a finite prefix of the run to
deduce its counter valuation. Let $A_c=(n,Q,\Sigma,Q_{0},\triangle,F)$
be a run-bounded weak B\"uchi counter automaton and $\rho$ be one of
its runs. We say that $\rho$ is an {\em accepting} run and that $w$ is
{\em accepted} by $A_c$ with the counter valuation $\vect{v}$ if and
only if $\infi(\rho)\cap (F\times {\lbrace}\vect{v}{\rbrace})\not=
\emptyset$, where $\infi(\rho)$ is the set of configurations that
appear infinitely often in $\rho$. Otherwise $\rho$ is {\em rejecting}
for $w$.\\
\newline
In the rest of this paper, we will only consider finite-word and
run-bounded weak B\"uchi counter automaton. We can now define a notion
of counter language, which takes the counters into account.

\begin{definition}
The {\em counter language} of a counter automaton $A_c$,
denoted ${\cal{L}}(A_c)$, is the set of pairs $(w,\vect{v})$ such that
$w$ can be accepted by $A_c$ with counter valuation $\vect{v}$.
\end{definition}

Observe that the class of counter-word automata is particular with
respect to existing classes of counter automata{\footnote{As an
    example, we cannot test the values of the counters.}} such as
reversal bounded counter automata \cite{Iba78}, constraint automata
\cite{HR98}, Parikh automata \cite{KR03}, or weighted automata
\cite{Moh03}. Indeed, counter-word automata use the counter part of
the automaton to assign counter valuations to a word when this word is
accepted by the automaton, rather than to restrict the language
accepted by the automaton. Introducing constraints on the counters
before the word is accepted{\footnote{As an example, one could
    associate constraints on each transition.}} generally leads to
more powerful models{\footnote{As an example, models that can
    recognize nonregular languages\,\cite{KR03}.}} for which most
problems are undecidable. The expressiveness of those models is not
needed for the practical applications we considered in the paper.

   
\subsection{Graph-Based Operations}

In this section, the operations of intersection and composition
defined for finite automata are extended to counter automata. We have
the following definitions.

\begin{definition}
Let $A_{c_1}=(n_1,\vect{c_1},Q_1,\Sigma,Q_{01},\triangle_1,F_1)$ and
$A_{c_2}=(n_2,\vect{c_2},Q_2,\Sigma,$\\$Q_{02},\triangle_2,F_2)$ be
two finite-word (respectively, run-bounded weak B\"uchi) counter automata. The
{\em counter-intersection} between $A_{c_1}$ and $A_{c_2}$, denoted
$A_{c_1}\cap_c A_{c_2}$, is the finite-word (respectively, run-bounded weak
B\"uchi) counter automaton $A_c=(n_1+n_2,\vect{c_1}\times
\vect{c_2},Q,\Sigma,Q_0,\triangle,F)$ with $L(A_c)=L(A_{c_1})\cap
L(A_{c_2})$ and ${\cal{L}}(A_c)={\lbrace}(w,\vect{v})\in
\Sigma^{\infty}\times \nats^{n_1+n_2}\mid (\exists (w,\vect{v_1})\in
      {\cal{L}}(A_{c_1}))(\exists (w,\vect{v_2})\in
      {\cal{L}}(A_{c_2}))(\vect{v}=\vect{v_1}\times
      \vect{v_2}){\rbrace}$.
\end{definition}

\begin{definition}
Let $T_{c_1}=(n_1,\vect{c_1},Q_1,\Sigma^2,Q_{01},\triangle_1,F_1)$ and
$T_{c_2}=(n_2,\vect{c_2},Q_2,\Sigma^2,$\\$Q_{02},\triangle_2,F_2)$ be
two finite-word (respectively, run-bounded weak B\"uchi) counter
transducers. The {\em counter-composition} of $T_{c_1}$ by $T_{c_2}$,
denoted $T_{c_2}\circ_c T_{c_1}$, is the finite-word
(respectively, run-bounded weak B\"uchi) counter transducer
$T_c=(n_1+n_2,\vect{c_1}{\times}\vect{c_2},Q,\Sigma^2, Q_0,\triangle,F)$,
with $L(T_c)=L(T_{2}\circ T_{1})$ and
${\cal{L}}(T_c)={\lbrace}(w,\vect{v})\in \Sigma^{\infty}\times
\nats^{n_1+n_2}\mid (\exists (w_1,\vect{v_1})\in
     {\cal{L}}(T_{c_1}))(\exists (w_2,\vect{v_2})\in
     {\cal{L}}(T_{c_2}))(\vect{v}=\vect{v_1}\times \vect{v_2}\wedge
     w=w_2\circ w_1){\rbrace}$.
\end{definition}

\begin{definition}
Let $T_1=(Q_1,\Sigma^2,Q_{01},\triangle_1,F_1)$ be a finite-word
(respectively, run-bounded weak B\"uchi) transducer, and
$A_{c_2}=(n_2,\vect{c_2},Q_2,\Sigma,Q_{02},\triangle_2,F_2)$ be a
finite-word (respectively, run-bounded weak B\"uchi) counter
automaton. The {\em counter-image} of $A_{c_2}$ by $T_1$, denoted
$T_1(A_{c_2})$, is the finite-word (respectively, B\"uchi) counter
automaton $A_{c}=(n_2,\vect{c_2},Q,\Sigma,Q_0,\triangle,F)$, where
$L(A_{c})=L(T_{1}(A_{c_2}))$ and
${\cal{L}}(A_{c})={\lbrace}(w,\vect{v_2})\in \Sigma^{\infty}\times
\nats^{n_2}\mid (\exists w_1\in L(T_1))(\exists (w_2,\vect{v_2})\in
     {\cal{L}}(A_{c_2}))(w=w_2\circ w_1){\rbrace}$.
\end{definition}

\subsection{Counter-Based Operations}
\label{counter-manipulations}

Let $A_c$ be a n-dimensional counter automaton over the alphabet
$\Sigma$, and $d$ its maximal increment value. The {\em extended
automaton} of $A_c$, denoted $(A_c)^e$, is the finite automaton
(without counters) obtained from $A_c$ by augmenting the label of each
of its transitions with its corresponding counter valuation. We have
the following definition.

\begin{definition}
Let $A_c=(n,\vect{c},Q,\Sigma,Q_0,\triangle,F)$ be a counter automaton
whose maximal increment value is $d$. The {\em extended} automaton
corresponding to $A_c$ is the finite automaton
$A=(Q,\Sigma',Q_0,\triangle',F)$, where
\begin{itemize}
\item
$\Sigma'=\Sigma\times {\lbrack}0,d{\rbrack}^n$, and 
\item
$\triangle'={\lbrace}(q,a',q')\in Q'\times \Sigma'\times Q'\mid
(\exists \vect{v}\in \nats^n)((q,(a,\vect{v}),q')\in
\triangle\wedge a'=a\times \vect{v}){\rbrace}$.
\end{itemize}
\end{definition} 

\noindent
A n-dimensional counter automaton over an alphabet $\Sigma$ and whose
maximal increment value is $d$ can be viewed as a finite automaton
over an alphabet $\Sigma\times {\lbrack}0,d{\rbrack}^n$ and,
alternatively, a finite automaton over an alphabet $\Sigma\times
{\lbrack}0,d{\rbrack}^n$ can be viewed as a n-dimensional counter
automaton over an alphabet $\Sigma$ and whose maximal increment value
is $d$. The alphabet $\Sigma\times {\lbrack}0,d{\rbrack}^n$ is
referred to as the {\em extended alphabet} of $A_c$.\\
\newline
If $A_c$ is a finite-word counter automaton, then we say that it is
{\em universal} if and only if $L((A_c)^e)=(\Sigma\times
{\lbrack}0,d{\rbrack}^n)^*$. If $A_c$ is a run-bounded weak B\"uchi
counter automaton, then it is {\em universal} if and only if $L((A_c)^e)=
(\Sigma\times {\lbrack}0,d{\rbrack}^n)^*(\Sigma\times 0)^\omega$.

\begin{definition}
Consider two counter automata $A_{c_1}$ and $A_{c_2}$ of same
dimensions. The {\em extended intersection (respectively, union)}
between $A_{c_1}$ and $A_{c_2}$, denoted $A_{c_1}\cap_e A_{c_2}$
(respectively, $A_{c_1}\cup_e A_{c_2}$), is a counter automaton $A_c$
such that $(A_c)^e=(A_{c_1})^e\cap (A_{c_2})^e$ (respectively,
$(A_c)^e=(A_{c_1})^e\cup (A_{c_2})^e$).
\end{definition}

The extended intersection (respectively, union) of two counter automata
can easily be computed by applying a classical intersection
(respectively, union) algorithm to their extended version. We also have
the following proposition.

\begin{proposition}
The extended intersection/union of two run-bounded weak B\"uchi counter
automata is a run-bounded weak B\"uchi counter automaton.
\end{proposition}

\begin{definition}
Let $A=(Q,\Sigma,Q_0,\triangle,F)$ be a finite-word (respectively,
B\"uchi automaton), the {\em counter-zero automaton} corresponding to
$A$ is the one-dimensional counter automaton
$A_c=(1,{\vect{c_1}},Q,\Sigma,Q_0,\triangle',F)$, where
\begin{itemize}
\itemsep0cm
\item
$\triangle'={\lbrace}(q,(a,\vect{0}),q')\in Q\times (\Sigma\times
\vect{0})\times Q\mid (q,a,q')\in \triangle{\rbrace}$.
\end{itemize}
\end{definition}


The problem of testing the equivalence between counter languages is
known to be undecidable for many classes of counter automata
\cite{Iba78}, but decidability results exist for some very particular
classes~\cite{Roos88}. The algorithms involved in those decidability
results are known to be of high complexity and difficult to
implement. Rather than trying to extend those results to counter-word
automata, we preferred to propose a sufficient criterion that can
easily be implemented with simple automata-based manipulations. Our
criterion is formalized with the following proposition.

\begin{proposition}
\label{th-counter1}
Let $A_{c_1}$ and $A_{c_2}$ be two finite-word (respectively, B\"uchi)
counter automata of same dimension. If $L(A_{c_1}^e)=L(A_{c_2}^e)$,
then ${\cal{L}}(A_{c_1})={\cal{L}}(A_{c_2})$.
\end{proposition}

\noindent
There are of situations where $L(A_{c_1}^e)\not=L(A_{c_2}^e)$,
while ${\cal{L}}(A_{c_1})={\cal{L}}(A_{c_2})$.

\begin{figure}[t]
\centering 

\subfigure[$A_1$]{
\label{ex-counter1}
\includegraphics[width=6.5cm]{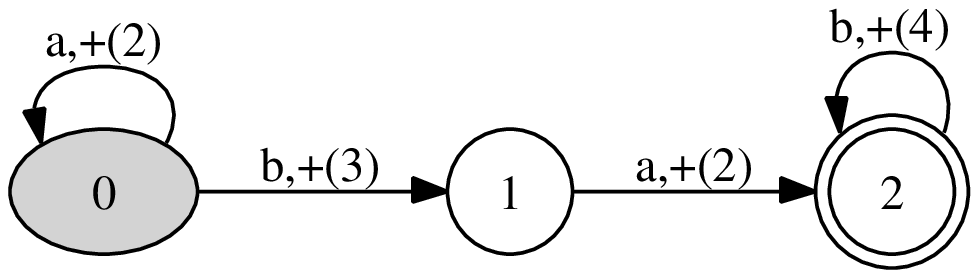}
\label{ex-counter1}
}
\hspace{2em}
\subfigure[$A_2$]{
\label{ex-counter2}
\includegraphics[width=6.5cm]{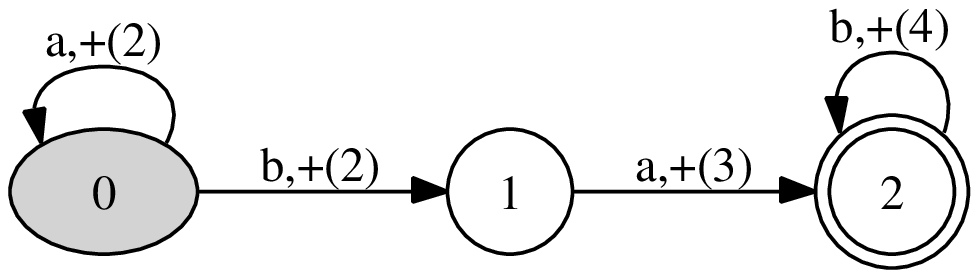}
\label{ex-counter2}
}
\caption{Two finite-word counter automata.}
\label{counter-extend}
\end{figure}

\begin{example}
Consider the two finite-word counter automata $A_{c_1}$ and $A_{c_2}$
given in Figure \ref{counter-extend}. The automaton $A_{c_1}^e$ does
not accept the same language as $A_{c_2}^e$. However
${\cal{L}}(A_{c_1})={\cal{L}}(A_{c_2})$.
\end{example}

The projection operation for finite automata extends to a {\em counter
  projection} for counter automata. We have the following definition.

\begin{definition}
\label{counter-projection}
Let $A_c=(n,\vect{c},Q,\Sigma, Q_0, \triangle,F)$ be a counter
automaton. For $1{\leq}i{\leq}n$, the projection of $A_c$
w.r.t. counter $c_i$, denoted $\Pi_{\not= c_i}(A_c)$ is the counter
automaton $A_c'=(n-1,\vect{c'},Q,\Sigma, Q_0, \triangle',F)$, where
$\vect{c'}=(c_1,{\dots},c_{i-1},c_{i+1},{\dots},c_n)$,
$L(A_c)=L(A_c')$, and
${\cal{L}}(A_c')={\lbrace}(w,\vect{c_1}{\times}\vect{c_2})\in
\Sigma^{\infty}\times \nats^{n-1}\mid (\exists\,\vect{c_3}\in
\nats)$\\$((w,\vect{c_1}{\times}\vect{c_3}{\times}\vect{c_2})\in
     {\cal{L}}(A_c)){\rbrace}$.
\end{definition}
\noindent
In the rest of the paper, we use the shortcut
$\Pi_{(\not={\lbrace}c_1,c_2,\dots,c_n{\rbrace})}(A)$
for\\ $\Pi_{(\not= c_1)}(\Pi_{(\not= c_2)}\dots(\Pi_{(\not=
  c_n)}(A))\dots)$.\\
\newline
We now present a methodology that given a counter automaton $A$,
computes another counter automaton $A'$ whose accepting words are
those of $A$ that satisfy counter constraints. We start with the
following definition.


\begin{definition}
Let $A_c$ be a finite-word (respectively, run-bounded weak B\"uchi)
n-dimensional counter automaton and $1\,{\leq}\,i,j\,{\leq}n$ be an
integer. We define $(A_{c})^{c_i>c_j}$ to be the counter automaton
obtained from $A_c$ by removing all the accepting runs that do not
assign a greater value to $c_i$ than to $c_j$. The automaton
$(A_{c})^{c_i>c_j}$ may have an infinite set of states since its
language may not be regular.
\end{definition}

In the rest of the paper, we use the notation
$(A_c)^{(c_1>{\lbrace}c_2,\dots,c_n{\rbrace})}$ to denote\\ $(\dots
((A_c)^{c_1>c_2})^{c_1>c_3}\dots)^{c_1>c_n}$.\\
\newline
Let $A_c$ be a finite-word (respectively, run-bounded weak B\"uchi)
n-dimensional counter automaton over $\Sigma$ and whose maximal
increment value is $d$. . A way to compute $(A_c)^{c_1>c_2}$ could be
to build a universal finite-word (respectively, run-bounded weak
B\"uchi automaton) $A^U$ defined over the same extended alphabet as
$A_c$ and then take the extended intersection between
$(A^U)^{c_i>c_j}$ and $A^c$. For any word $w\in \Sigma^*$
(respectively, $w\in \Sigma^{\omega}$), the automaton
$(A^U)^{c_i>c_j}$ contains all the accepting runs on $w$ that satisfy
the condition $c_i>c_j$. Hence, taking the extended intersection
between $(A^U)^{c_i>c_j}$ and $A_c$ will remove from $A_c$ all the
accepting runs that do not satisfy $c_i>c_j$. However, since there is
no bound on the difference between the values of $c_i$ and $c_j$
before the word is accepted, the automaton $(A^U)^{c_i>c_j}$ will have
an infinite number of states. Indeed, there should be one state for
each possible value of $c_i-c_j$. To avoid having to working with
infinite-state automata, we impose a synchronization between the
counters that need to be compared. As a consequence, we may not
exactly compute $(A_{c})^{c_i>c_j}$, but an automaton whose language
and counter language are subsets of those of $(A_{c})^{c_i>c_j}$. As
we will see in Section \ref{section-precision}, imposing this
synchronization is sufficient for the applications we will
consider. We have the following definition.

\begin{definition}
\label{synchro-count}
Let $A_c=(n,\vect{c},Q,\Sigma, Q_0, \triangle,F)$ be a finite-word
(respectively, run-bounded weak B\"uchi) counter automaton and a {\em
  synchronization bound} $M \in \nats$.  Let $\Delta c_l(\sigma)$
denotes the difference between the value associated to the counter
$c_l$ in the last and in the first state of the subrun $\sigma$ of a
run $\rho$ on $w$.  The automaton $A_c$ is {\em $M$-synchronized} with
respect to the counters $c_{i}$ and $c_{j}$ if
$L(A_c)=L(A_c)^{c_i>c_j}$, and for each $w \in L(A_c)$ and each
accepting run $\rho$ on $w$, we have $\|\Delta c_j(\sigma)-\Delta
c_i(\sigma)\|$ ${\leq}$ $M$.
\end{definition}

\begin{definition}
\label{uni-counter-final}
The finite-word (respectively, run-bounded weak) counter automaton
$A^{MU}=(n,\vect{c},Q,\Sigma, Q_0, \triangle,F)$ is
$M-$Universal-synchronized w.r.t. counters $c_i$ and $c_j$ if and only
if it is $M-$synchronized w.r.t. $c_i$ and $c_j$, and
$L(A^{MU})=\Sigma^*$ (respectively, $L(A^{MU})=\Sigma^\omega$).
\end{definition}

\noindent
Rather than computing $(A_{c})^{c_i>c_j}$, we propose to compute a
$M$-synchronized automaton whose language and counter language are
subsets of those of $(A_{c})^{c_i>c_j}$. For this, we intersect
$A_{c}$ with a {\em $M$-Universal-synchronized automaton}. Observe
that we can have a possibly infinite number of automata which are
$M-$Universal-synchronized w.r.t. $c_i$ and $c_j$. Clearly, when
taking the extended intersection between a counter automaton $A_c$ and
a $M$-Universal-synchronized automaton $A^{MU}$ defined over the same
extended alphabet, we obtain an automaton which is $M$-synchronized
and whose language and counter language are subsets of those of
$A_c$. The requirement $L(A^{MU})=\Sigma^*$ (respectively,
$L(A^{MU})=\Sigma^\omega$) in Definition \ref{uni-counter-final} is to
make sure that accepting runs are removed from $A_c$ only if they do
not satisfy the constraints over $c_i$ and $c_j$.


\section{The ($\omega$)-Regular Model Checking Framework}
\label{paper-model}

In this paper, we suppose that states of a system are encoded by words
over a fixed alphabet. If the states are encoded by finite words, then
sets of states can be represented by finite-word automata and
relations between states by finite-word transducers. This setting is
referred to as {\em Regular Model Checking}\,\cite{KMMPS97,WB98}. If
the states are encoded by infinite words, then sets of states can be
represented by deterministic weak B\"uchi automata and relations
between states by deterministic weak B\"uchi transducers. This setting
is referred to as {\em $\omega$-Regular Model
  Checking}\,\cite{BLW04a}. Formally, a finite automata-based
representation of a system can be defined as follows.\\
\newline
\begin{definition}
A {\em ($\omega$-)regular system} for a system $\TS=(S,S_0,R)$ is a
triple $M=(\Sigma, A, T)$, where
\begin{itemize}
\itemsep0cm
\item
$\Sigma$ is a finite alphabet over which the states are encoded as
finite (respectively, infinite) words;
\item
$A$ is a deterministic finite-word (respectively, deterministic weak
  B\"uchi) automaton over $\Sigma$ that represents $S_0$;
\item
$T$ is a deterministic finite-word (respectively, deterministic weak
  B\"uchi) transducer over $\Sigma^2$ that represents $R$. In the rest
  of the paper, $T$ is assumed to be reflexive.
\end{itemize}
\end{definition}

In the finite-word case, an execution of the system is an infinite
sequence of same-length finite words over $\Sigma$. The Regular Model
Checking framework was first used to represent parametric systems
\cite{AJMd02,BT02,KMMPS97,ABJN99,BJNT00,KPSZ02}. The framework can
also be used to represent various other models, which includes linear
integer systems\,\cite{WB95,WB00}, FIFO-queues systems\,\cite{BG96},
{\em XML} specifications\,\cite{BHRV06b,Td06}, and heap
analysis\,\cite{BHMV05,BHRV06b}.\\
\newline
As an illustration we give details on how to represent parametric
systems. Let $P$ be a process represented by a finite-state system. A
parametric system for $P$ is an infinite family
$S={\lbrace}S_n{\rbrace}_{n=0}^{\infty}$ of networks where for a fixed
$n$, $S_n$ is an {\em instance} of $S$, {\emph{i.e.,}} a network
  composed of $n$ copies of $P$ that work together in parallel. In the
  Regular Model Checking framework, the finite set of states of each
  process is represented as an alphabet $\Sigma$. Each state of an
  instance of the system can then be encoded as a finite word
  $w=w(0){\dots}w(n-1)$ over $\Sigma$, where $w(i-1)$ encodes the
  current state of the $i$th copy of $P$. Sets of states of several
  instances can thus be represented by finite-word automata. Observe
  that the states of an instance $S_n$ are all encoded with words of
  the same length. Consequently, relations between states in $S_n$ can
  be represented by binary finite-word relations, and eventually by
  transducers.

\begin{example}
\label{example-token}
Consider a simple example of parametric network of identical processes
implementing a token ring algorithm. Each of these processes can be
either in idle or in critical mode, depending on whether or not it
owns the unique token. Two neighboring processes can communicate with
each other as follows: a process owning the token can give it to its
right-hand neighbor. We consider the alphabet
$\Sigma={\lbrace}N,T{\rbrace}$. Each process can be in one of the two
following states : $T$ (has the token) or $N$ (does not have the
token). Given a word $w\in \Sigma^*$ with $|w|=n$ (meaning that $n$
processes are involved in the execution), we assume that the process
whose states are encoded in position $w(0)$ is the right-hand neighbor
of the one whose states are encoded in position $w(n-1)$. The
transition relation can be encoded as the union of two regular
relations that are the following:
\begin{enumerate}
\item
$(N,N)^*(T,N)(N,T)(N,N)^*$ to describe the move of the token from
  $w(i)$ to $w(i+1)$ (with $0{\leq}i{\leq}n-2$), and
\item
$(N,T)(N,N)^*(T,N)$ to describe the move of the token from
$w(n-1)$ to $w(0)$.
\end{enumerate}

The set of all possible initial states where the first process has the
token is given by $TN^*$. 
\end{example}

In the infinite-word case, an execution of the system is an infinite
sequence of infinite words over $\Sigma$. The $\omega$-Regular Model
Checking framework has been used for handling systems
with both integer and real variables~\cite{BW02,BJW05}, such as
linear hybrid systems with a constant derivative (see examples in
\cite{ACHH95} or in \cite{BLW04b,Leg07}).\\
\newline
It is known that verifying properties of systems in the
($\omega$-)Regular Model Checking framework generally reduces to
solving the {\em ($\omega$-)Regular Reachability
  Problems}\,\cite{PS00,BJNT00,BLW04a,AJNdS04,Leg07,BLW04b} that are
defined hereafter.

\begin{definition}  
Let $A$ be a deterministic finite-word (respectively, deterministic
weak B\"uchi) automaton, and $T$ be a reflexive deterministic
finite-word (respectively, deterministic weak B\"uchi) transducer. The
($\omega$-)Regular Reachability Problems for $A$ and $T$ are the
following:
\begin{enumerate}
\item
{\em Computing $T^*(A)$:} the goal is to compute a finite-word
(respectively, weak B\"uchi) automaton representing $T^*(A)$. If $A$
represents a set of states $S$ and $T$ a relation $R$, then $T^*(A)$
represents the set of states that can be reached from $S$ by applying
$R$ an arbitrary number of times;
\item
{\em Computing $T^*$:} the goal is to compute a finite-word
(resp. weak B\"uchi) transducer representing the reflexive transitive
closure of $T$. If $T$ represents a subset of a power of a
reachability relation $R$, then $T^*$ represents its closure.
\end{enumerate}
\end{definition}

The ($\omega$-)Regular Reachability Problems are
undecidable\,\cite{AK86}, but partial solutions exist. Studying those
solutions is the subject of the rest of this paper.

\section{On Solving ($\omega$-)Regular Reachability Problems}
\label{section-solving}

Among the techniques to solve the ($\omega$-)Regular Reachability
Problems, one distinguishes between {\em domain specific and generic}
techniques. Domain specific techniques exploit the specific properties
and representations of the domain being considered and were for
instance obtained for systems with FIFO-queues in~\cite{BG96,BH97},
for systems with integers and reals in~\cite{Boi98,BW02,BHJ03}, for
pushdown systems in~\cite{FWW97,BEM97}, and for lossy queues
in~\cite{AJ96}.  Generic
techniques\,\cite{KMMPS97,BJNT00,JN00,BHV04,BLW03,BLW04a,Tou01,DLS02,AJNd03,VSVA04,VSVA05}
consider automata-based representations and provide algorithms that
operate directly on these representations, mostly disregarding the
domain for which it is used.\\
\newline
In this paper, we propose a new generic technique to solving the
($\omega$-)Reachability Problems. We use the following definition.

\begin{definition}
Given a possibly infinite sequence $A^1, A^2, \dots$ of automata, the
limit of this sequence is an automaton $A^*$ such that $L(A^*)=\bigcup
L(A^i)$.
\end{definition}

\noindent
Consider a transducer $T$ and an automaton $A$. We first observe that
the computations of both $T^*$ and $T^*(A)$ can be reduced to the
computation of the limit of a possibly infinite sequence of
automata. Indeed, computing $T^*$ amounts to compute the limit of
$T_{id}$, $T^{1}$, $T^{2}$, $T^{3}$, \ldots, and computing $T^*(A)$
amounts to compute the limit of $A$, $T^{1}(A)$, $T^{2}(A)$,
$T^{3}(A)$, \ldots. We propose a generic technique which can compute
the limit of a sequence of automata by extrapolating one of its finite
{\em sampling sequence}, {\emph{i.e.}}  selected automata from a finite
  prefix of the sequence. The extrapolation step proceeds by comparing
  successive automata in the sampling sequence, trying to identify the
  difference between these in the form of an {\em increment}, and {\em
    extrapolating} the repetition of this increment by adding loops to
  the last automaton of the sequence. After the extrapolation has been
  built, one has to check whether it corresponds to the limit of the
  sequence. If this is the case, the computation terminates,
  otherwise, another sampling sequence has to be chosen. This is a
  semi-algorithm since there is no guarantee that (1) we can find a
  sampling sequence that can be extrapolated, and (2) the result of
  the extrapolation will be the desired closure.\\
\newline
The presentation of our solution is organized as follows. Section
\ref{section-sampling} discusses the choice of the sampling
sequence. Section \ref{section-increment} presents a methodology to
detect increments. Section \ref{section-extrapolation} presents
several extrapolation algorithms. Finally, Section
\ref{section-precision} introduces criteria to determine the
correctness of the extrapolation. An implementation of those results
as well as some experiments are presented in Section
\ref{section-implementation}. 

\section{Choosing the Sampling Sequence}
\label{section-sampling}

Choosing the sampling sequence is a rather tricky issue and there is
no guarantee that this can be done in a way that ensures that the
extrapolation step can be applied. However, there are heuristics that
are very effective for obtaining a sampling sequence that can be
extrapolated. The following lemma shows that the sampling sequence can
be selected quite arbitrarily, assuming that $T$ is {\em reflexive}.

\begin{lemma}
\label{lemma-sample1}
Let $T$ be a reflexive transducer and $A$ be an automaton. If
$s=s_0,s_1, s_2, \ldots$ is an infinite increasing subsequence of the
natural numbers, then $L(T^*) = \bigcup_{k \geq 0} L(T^{s_k})$ and,
similarly, $L(T^*(A))= \bigcup_{k \geq 0} L(T^{s_k} (A))$.
\end{lemma}

\begin{proof}
The lemma follows directly from the fact that for any $i\geq 0$, there
is an $s_k\in s$ such that $s_k > i$ and that, since $T$ is reflexive,
$(\forall j \leq i)(L(T^j) \subseteq L(T^i))$ (respectively,
$L(T^j(A)) \subseteq L(T^i(A))$).
\end{proof}

As an example, for the cases of FIFO-queue, pushdown, and parametric
systems, we observed that considering sample points of the form
$s_k=ak$, where $a\in \nats$ is a constant, turns out to be very
useful. For the case of arithmetic, we observed that the useful
sampling points are often of the form $s_k=a^k$. Sampling sequences
with sampling points of the form $s_k=ak$ are called {\em linear},
while sampling sequences with sampling points of the form $s_k=a^k$
are called {\em exponential}.

\begin{example}
\label{last-example} 

Figure~\ref{fig-xx1b} shows the minimal transducer of Example
\ref{ex-xx1} composed with itself $2$, $4$, $8$ and $16$ times. The
difference between the graphs for $T^4$ and $T^8$ takes the form of an
increment represented by the set of states ${\lbrace}2,6{\rbrace}$ in
$T^8$. This increment is repeated between $T^8$ and
$T^{16}$. Consequently, $T^{16}$ differs from $T^4$ by the addition of
two increments represented by the sets ${\lbrace}3,8{\rbrace}$ and
${\lbrace}2,7{\rbrace}$.
\begin{figure}
\begin{center}
\includegraphics{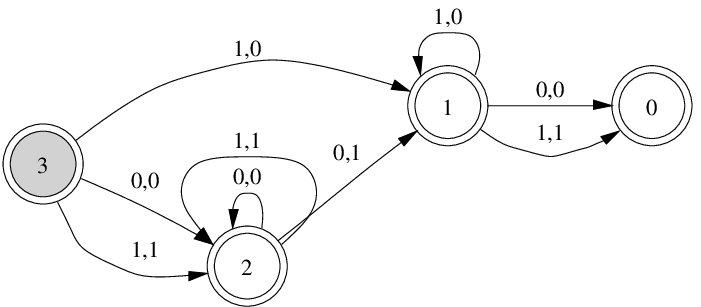}
\includegraphics{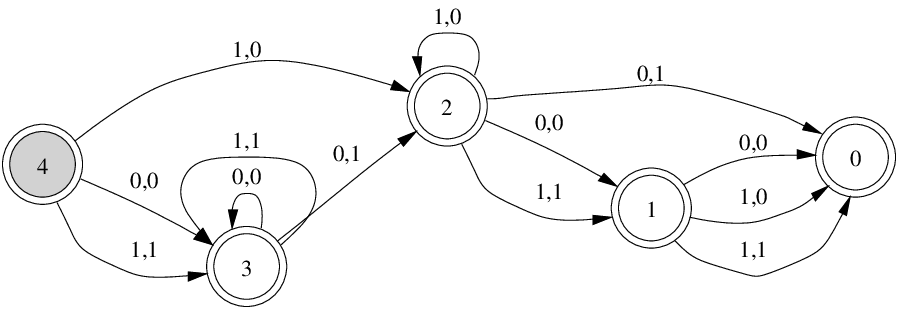}
\includegraphics{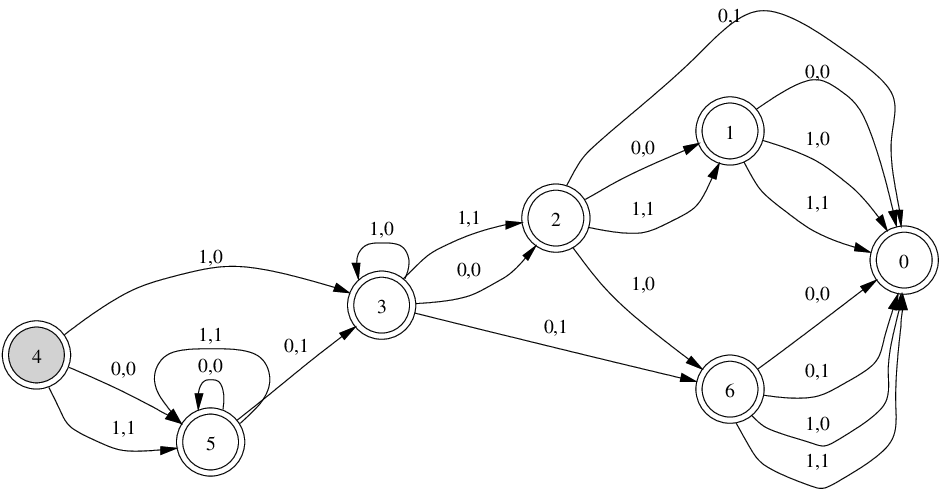}
\includegraphics{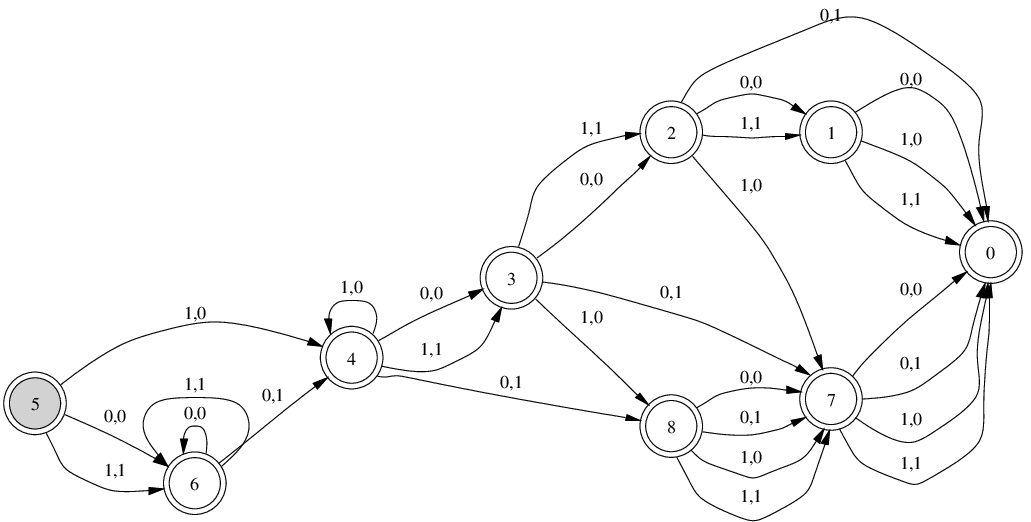}
\end{center}
\caption{Transducer of Example \ref{last-example} at powers of two.
\label{fig-xx1b}}
\end{figure}
\end{example}

\section{Detecting Increments}

\label{section-increment}

We consider a finite sequence $A^1, A^2, A^3, \ldots, A^n$ of finite
automata that are either all finite-word automata or all weak B\"uchi
automata. Those automata are assumed to be deterministic and
minimal. Our goal is to determine whether, for sufficiently large $i$,
the automaton $A^{i+1}$ differs from $A^i$ by some additional constant
finite-state structure. Our strategy, consists in comparing a finite
number of successive automata until a suitable increment can be
detected.

For each $i > 0$, let $A^{i} = (Q^{i}, \Sigma, q_0^{i}, \delta^{i},
F^{i})$. To identify common parts between two successive automata
$A^{{i}}$ and $A^{{i+1}}$ we first look for states of $A^{{i}}$ and
$A^{{i+1}}$ from which identical languages are accepted. Precisely, we
compute a forward equivalence relation $E^{i}_f \subseteq Q^{i} \times
Q^{{i+1}}$ between $A^{{i}}$ and $A^{{i+1}}$. Since we are dealing
with deterministic minimal automata, the forwards equivalence $E_f^i$
is one-to-one (though not total) and can easily be computed by
partitioning the states of the joint automaton $(Q^{i} \cup Q^{i+1},
\Sigma, q_{0}^i, \delta^i \cup \delta^{i+1}, F^{i} \cup F^{i+1})$
according to their accepted language. For finite-word automata, this
operation is easily carried out by Hopcroft's finite-state
minimization procedure~\cite{Hop71}. For weak B\"uchi automata, one
uses the variant introduced in \cite{Loe01}.

\begin{remark}
Note that because the automata are minimal, the parts of
$A^i$ and $A^{i+1}$ linked by $E_i^f$ are isomorphic (see Definition
\ref{iso-def}), incoming transitions being ignored.
\end{remark}

Next, we search for states of $A^{{i}}$ and $A^{{i+1}}$ that are
reachable from the initial state by identical languages.  Precisely,
we compute a backward equivalence relation $E^{i}_b \subseteq Q^{i}
\times Q^{{i+1}}$ between $A^{{i}}$ and $A^{{i+1}}$. Since $A^i$ and
$A^{i+1}$ are deterministic and minimal, the {\em backwards
  equivalence} $E_b^i$ can be computed by forward propagation,
starting from the pair $(q_{0}^i, q_{0}^{i+1})$ and exploring the
parts of the transition graphs of $A^i$ and $A^{i+1}$ that are
isomorphic to each other, if transitions leaving these parts are
ignored.

\begin{remark}
Note that because the automata are minimal, the parts of
$A^i$ and $A^{i+1}$ linked by $E_b^i$ are isomorphic, outgoing
transitions being ignored.
\end{remark}

We now define a notion of finite-state {\em increment} between two
successive automata, in terms of the relations $E^{i}_f$ and
$E^{i}_b$.

\begin{definition}
\label{def-incr-larger1}
Let $A^{i} = (Q^{i}, \Sigma, q_0^{i}, \delta^{i}, F^{i})$ and $A^{i+1}
= (Q^{i+1}, \Sigma, q_0^{i+1}, \delta^{i+1}, F^{i+1})$ be two minimal
finite-word (respectively, minimal weak B\"uchi) automata. Let $E^i_b$
and $E^i_f$ be respectively, the backward and forward equivalences
computed between $A^{i}$ and $A^{i+1}$.  The automaton $A^{i+1}$ is\/
{\em incrementally larger} than $A^i$ if the relations $E^i_f$ and
$E^i_b$ cover all the states of $A^i$. In other words, for each $q \in
Q^i$, there must exist $q' \in Q^{i+1}$ such that $(q, q') \in E^i_b
\cup E^i_f$.
\end{definition}

If $A^{i+1}$ is incrementally larger than $A^{i}$, the {\em increment}
consists of the states that are matched neither by $E^{i}_f$, nor by
$E^{i}_b$.

\begin{definition}
\label{def-incr-larger2}
Let $A^{i} = (Q^{i}, \Sigma, q_0^{i}, \delta^{i}, F^{i})$ and $A^{i+1}
= (Q^{i+1}, \Sigma, q_0^{i+1}, \delta^{i+1}, F^{i+1})$ be two minimal
finite-word (respectively, minimal weak B\"uchi) automata. Let $E^i_b$
and $E^i_f$ be respectively, the backward and forward equivalences
computed between $A^{i}$ and $A^{i+1}$. If $A^{i+1}$ is incrementally
larger than $A^{i}$, then
\begin{enumerate}
\item
the set $Q^{i}$ can be partitioned into $\{ Q_b^i, Q_f^i\}$, such that
\begin{itemize}
\item The set $Q_f^i$ contains the states $q$ covered by $E^i_f$,
i.e., for which there exists $q'$ such that $(q, q') \in E^i_f$;
\item The set $Q_b^i$ contains the remaining states.
\end{itemize}
\item
The set $Q^{i+1}$ can be partitioned into
$\{ Q_H^{i+1}, Q_{I_0}^{i+1}, Q_T^{i+1} \}$, where
\begin{itemize}
\item The\/ {\em head part} $Q_H^{i+1}$ is the image by $E^{i}_b$
of the set $Q_b^{i}$;
\item The\/ {\em tail part} $Q_T^{i+1}$ is the image by $E^{i}_f$
of the set $Q_f^{i}$, dismissing the states that belong to
$Q_H^{i+1}$ (the intention is to have an unmodified head part);
\item The\/ {\em increment} $Q_{I_0}^{i+1}$ contains the states
that do not belong to either $Q_H^{i+1}$ or $Q_T^{i+1}$.
\end{itemize}
\end{enumerate}
\end{definition}

\begin{figure}
\centerline{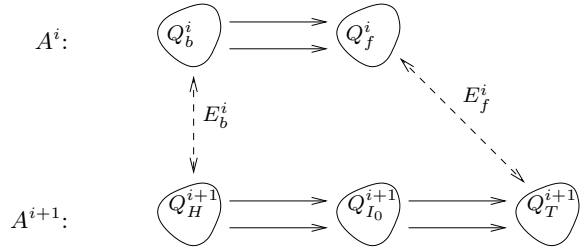}
\caption{Partitioning automata states.}
\label{fig-parts}
\end{figure}

\noindent
Definitions \ref{def-incr-larger1} and \ref{def-incr-larger2} are
illustrated in Figure \ref{fig-parts}.\\
\newline


Our expectation is that, when moving from one automaton to the next in
the sequence, the increment will always be the same.  We formalize
this property with the following definition.

\begin{definition}
\label{def-incr-seq}
Let $S_I=A^{i}$, $A^{i + 1}$, \ldots, $A^{i + k}$\/ and for each
$0\,{\leq}\,j\,{\leq}\,k$, let $A^{i+j} = (Q^{i+j}, \Sigma, q_0^{i+j},
\delta^{i+j}, F^{i+j})$ be a finite-word (respectively, weak B\"uchi)
automata. For each $0\,{\leq}\,j<k$, let $E_b^{i+j}$ and $E_f^{i+j}$
be respectively, the backward and the forward equivalences computed
between $A^{i+j}$ and $A^{i+j+1}$.  The sequence $S_I$ is an {\em
  incrementally growing} sequence if
\begin{itemize}
\item For each $0\,{\leq}\,j\,{\leq}\,k$, $A^{i+j}$ is minimal;
\item For each  $0\,{\leq}\,j\,{\leq}\,k-1$, $A^{i+j+1}$ is 
incrementally larger than $A^{i+j}$;
\item For each $1\,{\leq}\,j\,{\leq}\,k-1$, the {\em head} increment
$Q_{I_0}^{i+j+1}$, which is detected between $A^{i+j}$ and
$A^{i+j+1}$, is the image by $E^{i+j}_b$ of the increment
$Q_{I_0}^{i+j}$.
\end{itemize}
\end{definition}

\noindent
Consider a subsequence $S_I=A^{i}$, $A^{i + 1}$, \ldots, $A^{i + k}$
of $A^1, \dots, A^n$ that grows incrementally. For
$2{\leq}\,j\,{\leq}\,n$, the tail part $Q_T^{i+j}$ of $A^{i+j}$ will
then consist of $j-1$ copies of its head increment $Q_{I_{0}^{i+}}$
plus a part that we will name the {\em tail-end set\/}. Precisely,
$Q_T^{i+j}$ can be partitioned into $\{ Q_{I_1}^{i+j}, Q_{I_2}^{i+j},
\ldots, Q_{I_{j-1}}^{i+j}, Q_{T_f}^{i+j} \}$, where
\begin{itemize}
\item For each $1{\leq}\,\ell\,\leq j-1$, the {\em tail increment\/}
$Q_{I_\ell}^{i+j}$ is the image by the relation 
$E_f^{i+j-1} \circ E_f^{i+ j-2} \circ \cdots
\circ E_f^{i+j-\ell}$
of the {\em head increment} $Q_{I_0}^{i+j-\ell}$;
\item The {\em tail-end set\/} $ Q_{T_f}^{i+j}$ contains the
remaining elements of $Q_T^{i+j}$.
\end{itemize}

\noindent
Given an automaton $A^{i + j}$ in the sequence $S_I$, we define its
{\em growing decomposition} w.r.t. $S_I$, denoted
$\it{GROW}_{(S_I)}(A^{i+j})$, to be the ordered list
${\lbrace}Q_H^{i+j},{\lbrace}Q_{I_0}^{i + j}, \ldots,
$\\$Q_{I_{j-1}}^{i + j}{\rbrace},Q_T^{i+j}{\rbrace}$. It is easy to
see that the head increment $Q_{I_0}^{i + j}$ of $A^{i+j}$ and all its
tail increments $Q_{I_\ell}^{i+j}$, $\ell \in [1,j-1]$ appearing in
its tail part $Q_T^{i + j}$ are images of the head increment
$Q_{I_0}^{i+1}$ detected between $A^{i}$ and $A^{i+1}$ by a
combination of forward and backward equivalences. This observation
extends to all the automata in $S_I$. Consequently the transition
graphs internal{\footnote{The transition graph only contains
    transitions between states of the increment.}} to all increments
of all the automata in the sequence are isomorphic to that of
$Q_{I_0}^{i+1}$, and hence are isomorphic to each other. In the rest
of the thesis, this isomorphism relation between two increments is
called the {\em increment isomorphism relation}. Observe also that,
since we are working with minimal automata, for each $j \in [1,k-1]$
we have the following:
\begin{itemize}
\item
The head part $Q_{H}^{i+j+1}$ is the image by $E^{i+j}_b$ of the head
part $Q_{H}^{i+j}$. Consequently, the internal transition graphs of
the head parts of all the automata in the sequence $S_I$ are
isomorphic to each other. This isomorphism relation is called the {\em
head isomorphism relation};
\item
The tail-end set $Q_{T_f}^{i+j+1}$ is the image by $E_f^{i+j}$ of the
tail-end set $Q_{T_f}^{i+j}$. Consequently, the internal transition
graphs of the tail-end sets of all the automata in the sequence $S_I$
are isomorphic to each other. This isomorphism relation is called the
{\em tail-end set isomorphism relation}.
\end{itemize}
\noindent 
The situation is illustrated in Figure~\ref{figure2}.\\
\newline
\begin{figure}
\begin{center}
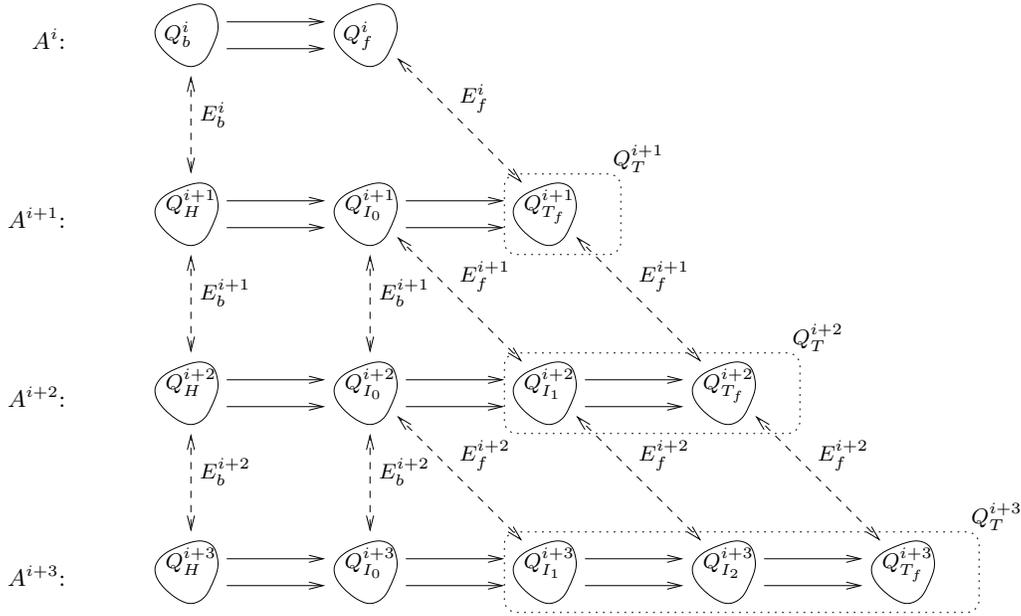
\end{center}
\caption{Automata in an incrementally growing sequence.}
\label{figure2}
\end{figure}
\noindent



Our intention is to extrapolate the last automaton of an incrementally
growing sequence of automata by adding more increments, following a
regular pattern. In order to do this, we need to compare and
characterize the transitions leaving different increments.

\begin{definition}
\label{communication-equiv}
Let $A^{i+k} = (Q^{i+k}, \Sigma, q_0^{i+k}, \delta^{i+k}, F^{i+k})$ be
the last automaton of an incrementally growing sequence of automata
$S_I=A^{i}$, $A^{i + 1}$, \ldots, $A^{i + k}$. Assume that
$\it{GROW}_{(S_I)}(A^{i+k})={\lbrace}Q_H^{i+k},{\lbrace}Q_{I_0}^{i +
k}, \ldots, Q_{I_{k-1}}^{i +
k}{\rbrace},Q_{T_f}^{i+k}{\rbrace}$. Then, an increment
$Q_{I_\alpha}^{i + k}$ ($0\,{\leq}\,\alpha\,\leq\,k-1$) is said to
be\/ {\em communication equivalent} to an increment $Q_{I_\beta}^{i +
k}$ ($0\,{\leq}\,\beta\,\leq\,k-1$) if and only if, for each pair of
corresponding states (by the increment isomorphism) $(q,q')$, $q \in
Q_{I_\alpha}^{i + k}$ and $q' \in Q_{I_\beta}^{i + k}$, and $a \in
\Sigma$, we have that, either
\begin{itemize}
\item
$\delta^{i+k}(q,a) \in Q_{I_\alpha}^{i + k}$ and $\delta^{i+k}(q',a)
\in Q_{I_\beta}^{i + k}$, hence leading to corresponding states by the
existing increment isomorphism between $Q_{I_\alpha}^{i + k}$ and
$Q_{I_\beta}^{i + k}$, or
\item
$\delta^{i+k}(q,a)$ and $\delta^{i+k}(q',a)$ are both undefined, or
\item
$\delta^{i+k}(q,a)$ and $\delta^{i+k}(q',a)$ both leading to the same
state of the tail end $Q_{T_f}^{i+k}$, or
\item
there exists some $\gamma>0$ such that $\delta^{i+k}(q,a)$ and
$\delta^{i+k}(q',a)$ lead to corresponding states by the increment
isomorphism between $Q_{I_{\alpha+\gamma}}^{i + k}$ and
$Q_{I_{\beta+\gamma}}^{i + k}$ ($0\,\leq\,\alpha+\gamma,
\beta+\gamma\,\leq\,k-1$).
\end{itemize}
\end{definition}

\noindent
The definition easily generalizes to increments of different automata.

\begin{figure}
\centerline{\includegraphics[width=8cm]{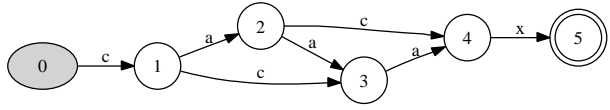}}
\caption{An automaton for Definition \ref{communication-equiv}.}
\label{increment-illustr1}
\end{figure}

\begin{example}
Consider the automaton of Figure \ref{increment-illustr1}, whose set
of states is given by ${\lbrace}0,1,2,3,4,5{\rbrace}$. Assume that $Q$
contains three increments that are $Q_{I_{0}}={\lbrace}1{\rbrace}$,
$Q_{I_{1}}={\lbrace}2{\rbrace}$, and
$Q_{I_{2}}={\lbrace}3{\rbrace}$. The increments $Q_{I_{0}}$ and
$Q_{I_1}$ are communication stable. The property does not hold for
$Q_{I_{0}}$ and $Q_{I_2}$ since a transition labeled with $c$ is not
defined from states $3$.
\end{example}

For the same reasons, we also need to compare the transitions leaving
the head part of different automata in the sequence.

\begin{definition}
\label{communication_stable}
Let $A^{i+k-1} = (Q^{i+k-1}, \Sigma, q_0^{i+k-1}, \delta^{i+k-1},
F^{i+k-1})$ and $A^{i+k} = (Q^{i+k}, \Sigma, q_0^{i+k}, \delta^{i+k},
F^{i+k})$ be the two last automata of an incrementally growing
sequence of automata $S_I=A^{i}$, $A^{i + 1}$, \ldots, $A^{i +
k}$. Assume that $\it{GROW}_{(S_I)}(A^{i+k-1})={\lbrace}Q_H^{i + k
-1},{\lbrace}Q_{I_0}^{i + k-1}, \ldots,Q_{I_{k-2}}^{i +
k-1}{\rbrace},Q_{T_f}^{i+k-1}{\rbrace}$ and
$\it{GROW}_{(S_I)}(A^{i+k})={\lbrace}Q_H^{i + k},{\lbrace}Q_{I_0}^{i +
k},$\\$ \ldots,Q_{I_{k-1}}^{i +
k}{\rbrace},Q_{T_f}^{i+k}{\rbrace}$. We say that $A^{i+k-1}$ and
$A^{i+k}$ are {\em communication stable} if and only if for each pair
of corresponding states (by the increment isomorphism) $(q, q')$, $q$
${\in}$ $Q_H^{i + k -1}$ and $q'$ ${\in}$ $Q_H^{i + k}$, and $a$
${\in}$ ${\Sigma}$, we have that, either

\begin{itemize}
\item 
$\delta^{i+k-1}(q,a)$ ${\in}$ $Q_H^{i+k-1} $ and $\delta^{i+k}(q',a)$
${\in}$ $Q_H^{i+k}$, hence leading to corresponding states by the
existing head isomorphism between $Q_H^{i+k-1}$ and $Q_H^{i+k}$, or
\item
$\delta^{i+k-1}(q,a)$ and $\delta^{i+k}(q',a)$ are both
undefined, or
\item
$\delta^{i+k-1}(q,a)= q_f^{i+k-1}\in Q_{T_f}^{i+k-1}$ and
$\delta^{i+k}(q',a)=q_f^{i+k}\in Q_{T_f}^{i+k}$, hence leading to
corresponding states by the existing tail-end set isomorphism between
$Q_{T_f}^{i+k-1}$ and $Q_{T_f}^{i+k}$, or
\item 
$\delta^{i+k-1}(q,a)\in Q_{I_x}^{i+k-1}$ and $\delta^{i+k}(q',a)\in
Q_{I_x}^{i+k}$, hence leading to corresponding states by the existing
increment isomorphism between $Q_{I_x}^{i+k-1}$ and $Q_{I_x}^{i+k}$
($0\,\leq\,x\,\leq\,k-1$).
\end{itemize} 
\end{definition}






\section{Extrapolation Algorithms}
\label{section-extrapolation}

To extrapolate a possibly infinite sequence of minimal finite-word
(respectively, minimal weak B\"uchi) automata $A^1, A^2,{\dots}$ we
try to extract and extrapolate one of its finite incrementally growing
sampling sequences $S_I=A^{s_0},\dots,A^{s_k}$. The ``candidate''
extrapolation for $A^1, A^2,{\dots}$ is then given by the {\em
  extrapolation} of the sequence $S_I$. Let $A^{e_0}=A^{s_k}$ be the
last automaton of $S_I$. In order to extrapolate $S_I$, we simply
insert an extra increment between the head part of $A^{e_0}$ and its
head increment $Q_{I_0}^{e_0}$, and define its outgoing transitions in
order to make this extra increment communication equivalent to
$Q_{I_0}^{e_0}$. By repeatedly applying this extrapolation step we
obtain an extrapolated {\em infinite} sequence of automata
$A^{e_0},A^{e_1},\dots$ which is assumed to be the infinite extension
of the sampling sequence $S_I$. Formally, the {\em extrapolated
  sequence} of {\em origin} $A^{e_0}$ is the infinite sequence of
minimal automata $A^{e_0},A^{e_1},\dots$ such that
\begin{itemize}
\item
For each $i\,{\geq}\,0$,
$A^{s_0},A^{s_1},\dots,A^{s_{k-1}},A^{e_0},A^{e_1},\dots,A^{e_i}$
grows incrementally;
\item
For each $i>0$, $A^{e_i}$ is communication stable with $A^{e_0}$;
\item
For each $i>0$, the head increment detected between $A^{e_{i-1}}$ and
$A^{e_i}$ is communication equivalent to $Q_{I_0}^{e_0}$.
\end{itemize}

\noindent
The limit $A^{e_*}$ of the extrapolated sequence of origin $A^{e_0}$
is thus an extrapolation of the limit of $A^1, A^2, \dots$. In this
section, we present procedures to build a finite representation for
$A^{e_*}$. For technical reasons, the cases of finite-word and weak
B\"uchi automata are considered separately.


\subsection{Finite-word Automata}

Assume $A^{e_0}$ to be a finite-word automaton. We propose to build a
finite representation of $A^{e_*}$ by adding to $A^{e_0}$ new
transitions that simulate the existence of additional increments.

Consider the automaton $A^{e_0}$ with
$\it{GROW}_{(S_I)}(A^{e_0})={\lbrace}Q_H^{e_0},{\lbrace}Q_{I_0}^{e_0},
\ldots,
Q_{I_{k-1}}^{e_0}{\rbrace},$\\$Q_{T_f}^{e_0}{\rbrace}$. Suppose the
existence of a transition labeled by $a$ from a state $x$ of
$Q_{I_0}^{e_0}$ to a state $x'$ of $Q_{I_3}^{e_0}$. Since, the
increment $Q_{I_0}^{e_1}$ added between $A^{e_0}$ and $A^{e_1}$ is
communication equivalent to $Q_{I_0}^{e_0}$, there must exist a
transition $t$ labeled by $a$ from the state isomorphic to $x$ in
$Q_{I_0}^{e_1}$ to the state isomorphic to $x'$ in
$Q_{I_2}^{e_1}$. Our construction simulates $t$ in $A^{e_0}$ by adding
a transition $t'$ labeled by $a$ from $x$ to the state isomorphic to
$x'$ in $Q_{I_2}^{e_0}$. This construction can be repeated for the
addition of a second increment. The simulation of ``more than two
increments'' is done by adding transitions between states of
$Q_{I_0}^{e_0}$. Due to the communication equivalence property, a
similar principle has to be applied for outgoing transitions from
$Q_H^{e_0}$. The situation is illustrated in Figure \ref{extra-algo1}
where a part of $A^{e_0}$ has been represented. The dashed transitions
in the figure are the transitions added during the extrapolation
process.

\begin{figure}
\begin{center}
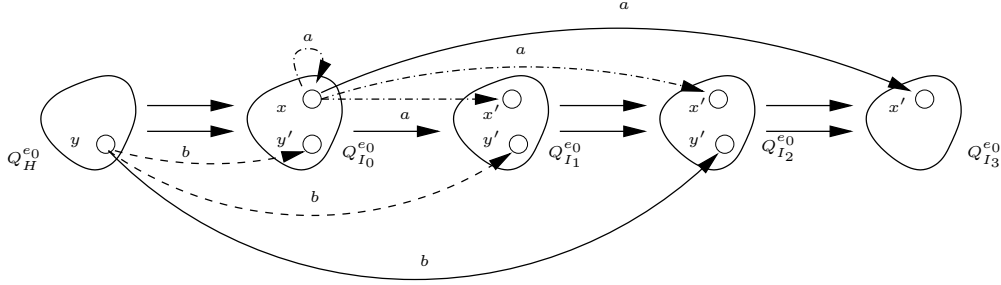
\end{center}
\caption{Illustration of the extrapolation procedure for finite-word
automata.
\label{extra-algo1}
}
\end{figure}

Formally, a finite representation of $A^{e_*}$ can be built from
$A^{e_0}$ with the construction underlined in the following
proposition.

\begin{proposition}
\label{extra-finite}
Let $A^{e_0}$ defined over $\Sigma$ be a minimal finite-word automaton
which is the last automaton of an incrementally growing sequence of
automata $S_I$.  Assume that
$\it{GROW}_{(S_I)}(A^{e_0})={\lbrace}Q_H^{e_0},{\lbrace}Q_{I_0}^{e_0},
\ldots,Q_{I_{k-1}}^{e_0}{\rbrace},Q_{T_f}^{e_0}{\rbrace}$. One can
compute a finite-word automaton $A^{e_*}$ that represents the limit of
the extrapolated sequence of origin $A^{e_0}$.
\end{proposition}

\begin{proof}
Let $\delta$ be the transition relation of $A^{e_0}$. The automaton
$A^{e_*}$ can be built from $A^{e_0}$ by augmenting $\delta$ using the
following rule:
\begin{quote}
For each state $q \in Q_H^{e_0}\cup Q_{I_0}^{e_0}$ and $a \in \Sigma$,
if $\delta(q,a)$ leads to a state $q'$ in an increment $Q_{I_{j}}{}$,
$1\leq j \leq k -1$, then for each $0\leq \ell < j$, add a transition
$(q,a,q'')$, where $q''$ is the state corresponding to $q'$ (by the
increment isomorphism) in $Q_{I_{\ell}}^{e_0}{}$.
\end{quote}
The added transitions, which include loops (transitions to
$Q_{I_0}^{e_0}$ itself) allow $A^{e_*}$ to simulate the runs of any of
the $A^{e_i}$ ($i\geq 0$). Conversely, it is also easy to see all
accepting runs generated using the added transitions correspond to
accepting runs of some $A^{e_i}$.
\end{proof}




\begin{figure}[t]
\centering 
\hspace{2em}
\subfigure[$A^{e_0}$]{
\label{extra4}
\includegraphics[width=8cm]{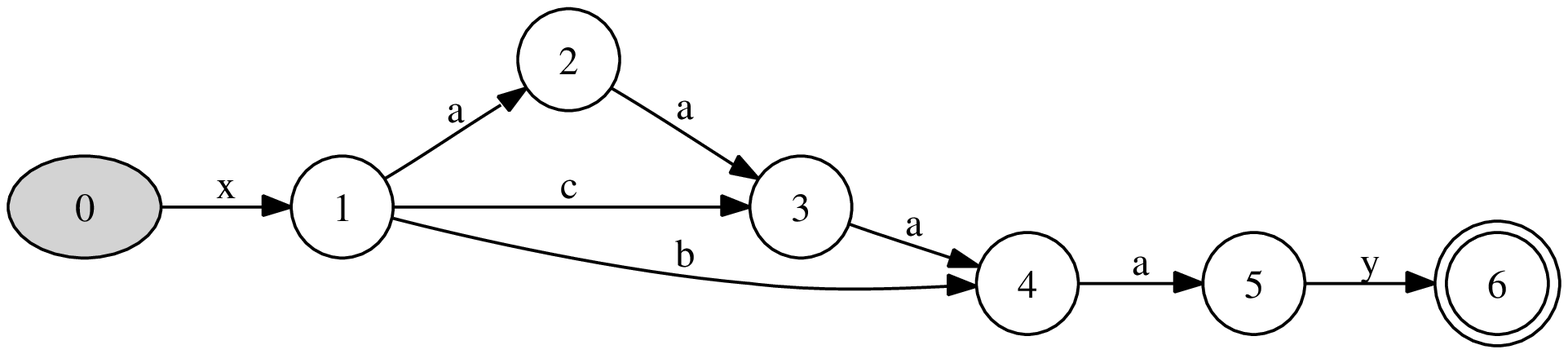}
\label{extra4}
}
\hspace{2em}
\subfigure[$A^{e_*}$]{
\label{extra5}
\includegraphics[width=8cm]{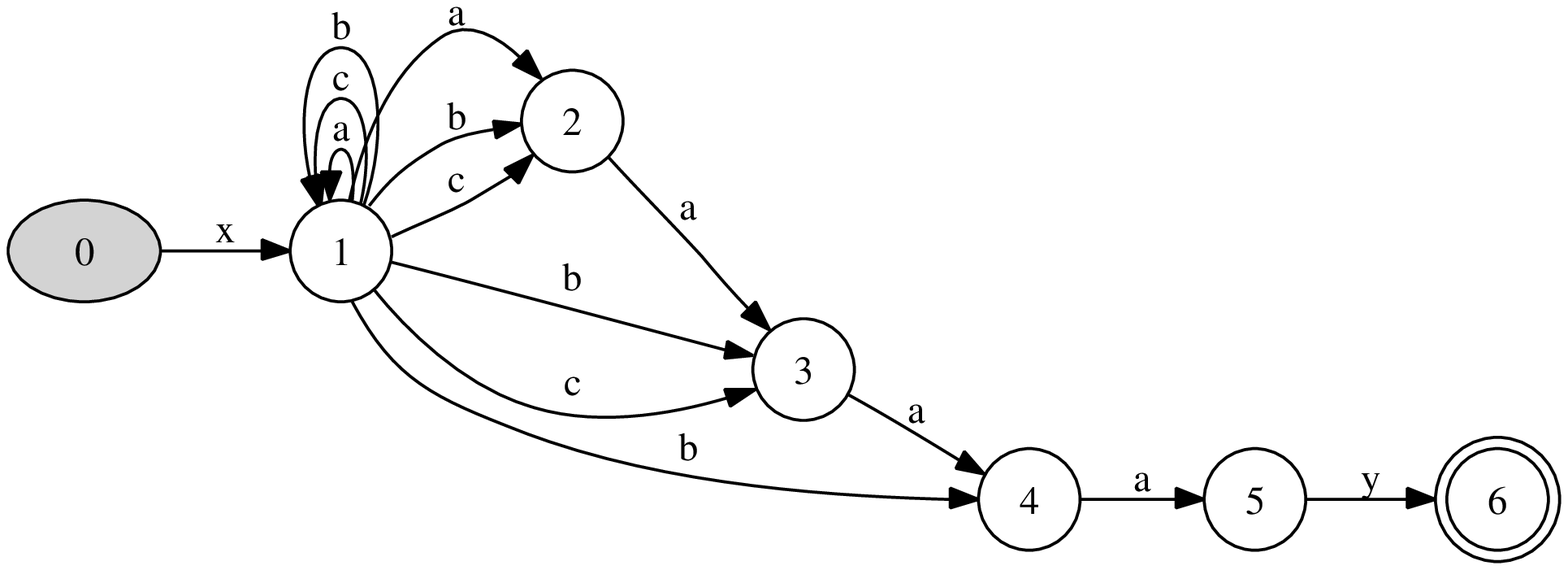}
\label{extra5}
}
\caption{Automata for Example \ref{example-extra1}.}
\end{figure}

\begin{example}
\label{example-extra1}
Consider the minimal finite-word automaton $A^{e_0}$ given in Figure
\ref{extra4}, with $Q_H^{e_0}={\lbrace}0{\rbrace}$,
$Q_{I_0}^{e_0}={\lbrace}1{\rbrace}$,
$Q_{I_1}^{e_0}={\lbrace}2{\rbrace}$,
$Q_{I_2}^{e_0}={\lbrace}3{\rbrace}$,
$Q_{I_3}^{e_0}={\lbrace}4{\rbrace}$, and
$Q_{T_f}^{e_0}={\lbrace}5,6{\rbrace}$. Applying the construction of
Proposition \ref{extra-finite} to $A^{e_0}$ gives the automaton
$A^{e_*}$ in Figure \ref{extra5}.
\end{example}

\begin{figure}
\begin{center}
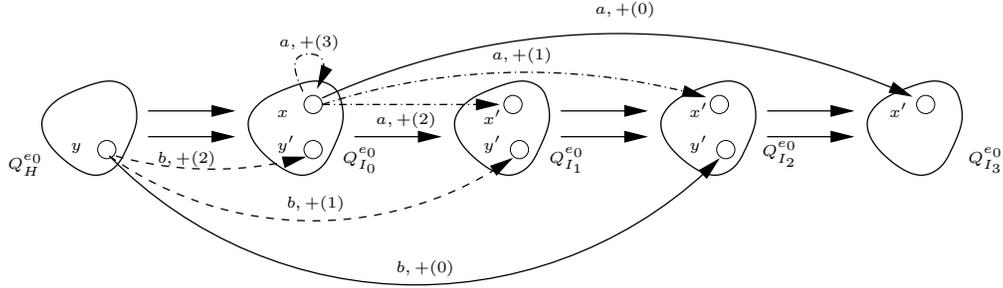
\end{center}
\caption{Illustration of the extrapolation algorithm for finite-word
automata with the addition of counter values.
\label{extra-algo2}
}
\end{figure}

We now show that it is possible to add a counter $c$ to $A^{e_*}$ in
such a way that when a word is accepted, the value of $c$ is the
smallest index $i$ of the automaton $A^{e_i}$ of the extrapolation
sequence by which the word is in fact accepted. Our construction
labels each transition added to $A^{e_0}$ with a value that represents
the number of increments simulated by this transition. In Figure
\ref{extra-algo2} we sketch the construction for the automaton given
in Figure \ref{extra-algo2}.

\begin{proposition}
\label{extra-finite-counter}
Let $A^{e_0}=(Q,\Sigma,Q_0,\delta,F)$ be a minimal finite-word
automaton which is the last automaton of a finite incrementally
growing sequence of automata $S_I$.  Assume that
$\it{GROW}_{(S_I)}(A^{e_0})={\lbrace}Q_H^{e_0},{\lbrace}Q_{I_0}^{e_0},
\ldots,Q_{I_{k-1}}^{e_0}{\rbrace},Q_{T_f}^{e_0}{\rbrace}$ and let
$A^{e_0}, A^{e_1},\dots$ be the extrapolated sequence of origin
$A^{e_0}$. One can compute a finite-word counter automaton $A^{e_*}_c$
such that (1) $L(A^{e_*}_c)={\bigcup}_{i{\geq}0}L(A^{e_i})$, (2) for
each $(w,i)$ $\in$ ${\cal{L}}(A^{e_*}_c)$, $w\in L(A^{e_i})$, and (3)
for each $i{\geq}0$, $w$ $\in$ $L(A^{e_i})$, $0{\leq}j{\leq}i$ exists
such that $(w,j)$ $\in$ ${\cal{L}}(A^{e_*}_c)$.
\end{proposition}

\begin{proof}
Let $\delta$ be the transition relation of $A^{e_0}$. The
one-dimensional counter automaton $A^{e_*}_c$ is given by
$(1,\vect{c},Q,\Sigma,Q_0,\triangle,F)$, with $\triangle$ defined as
follows:
\begin{itemize}
\item
Start with $\triangle={\lbrace}\emptyset{\rbrace}$;
\item
For each $(q,a,q')\in \delta$, add $(q,(a,\vect{0}),q')$ to $\triangle$;
\item
For each state $q \in Q_H^{e_0}\cup Q_{I_0}^{e_0}$ and $a \in
\Sigma$,
\begin{quote}
If $\delta(q,a)$ leads to a state $q'$ in an increment $Q_{I_{j}}{}$,
$1\leq j \leq k -1$, then for each $0\leq \ell < j$, add to
$\triangle$ a transition $(q,(a,j-l),q'')$, where $q''$ is the state
corresponding to $q'$ (by the increment isomorphism) in
$Q_{I_{\ell}}^{e_0}{}$.
\end{quote}
\end{itemize}
\end{proof}

\noindent
Let $A^{e_0}_c$ be the counter-zero automaton corresponding to
$A^{e_0}$. We directly see that for each $i>0$, $w$ $\in$
$L(A^{e_i})\setminus L(A^{e_0})$, $1\,{\leq}\,j\,{\leq}\,i$ exists
such that $(w,j)$ $\in$ ${\cal{L}}(A^{e_*}_c)\setminus
{\cal{L}}(A^{e_0}_c)$. Indeed, since $w\notin L(A^{e_0})$, any
accepted run on $w$ must pass by states of one of the added increments
and $j$ cannot be equal to $0$.


\begin{figure}[t]
\centering 
\includegraphics[width=10cm]{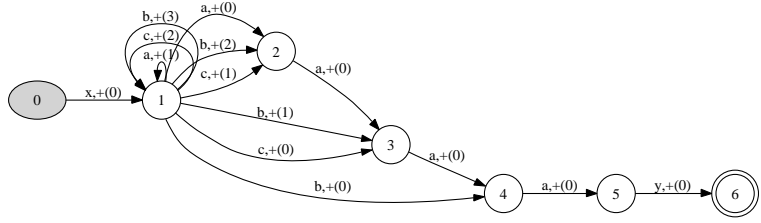}
\label{ex-counter0}
\caption{Automaton for Example \ref{example-extra2}.}
\label{extra6}
\end{figure}

\begin{example}
\label{example-extra2}
Figure \ref{extra6} presents the result of applying the construction
of Proposition \ref{extra-finite-counter} to Automaton $A^{e_0}$ of
Example \ref{example-extra1}.
\end{example}

\subsection{Weak B\"uchi Automata}
\label{extraproc-weak}

Assume now $A^{e_0}$ to be a deterministic weak B\"uchi automaton. In
such a case, a finite representation of the extrapolated sequence of
origin $A^{e_0}$ cannot be computed with the construction of
Proposition \ref{extra-finite}.

\begin{figure}[t]
\centering 
\subfigure[$A^{e_0}$]{
\label{counter-extra1}
\includegraphics[width=6.5cm]{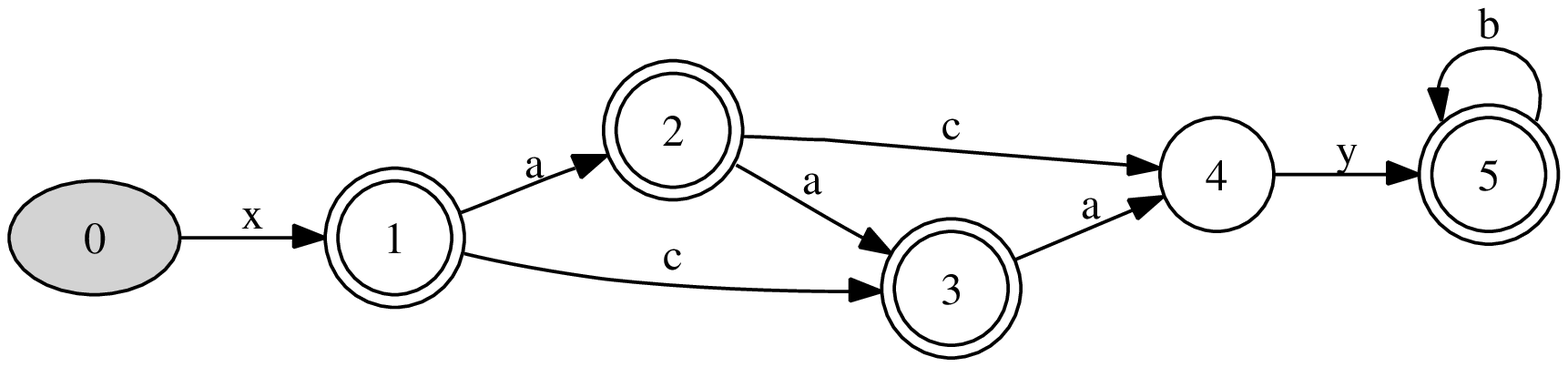}
\label{counter-extra1}
}
\hspace{2em}
\subfigure[$A^{e_*}$]{
\label{counter-extra2}
\includegraphics[width=6.5cm]{figs/extra2.ps}
\label{counter-extra2}
}
\caption{A weak B\"uchi automaton and its extrapolation with the
  construction of Proposition \ref{extra-finite} .}
\label{ex-counters2}
\end{figure}

\begin{example}
\label{example-extrapolation1}
Consider the minimal weak B\"uchi automaton $A^{e_0}$ given in Figure
\ref{counter-extra1}, with $Q_H^{e_0}={\lbrace}0{\rbrace}$,
$Q_{I_0}^{e_0}={\lbrace}1{\rbrace}$,
$Q_{I_1}^{e_0}={\lbrace}2{\rbrace}$,
$Q_{I_2}^{e_0}={\lbrace}3{\rbrace}$, and
$Q_{T_f}^{e_0}={\lbrace}4,5{\rbrace}$.  Applying the construction of
Proposition \ref{extra-finite} to $A^{e_0}$ gives the automaton
$A^{e_*}$ in Figure \ref{counter-extra2}. This automaton accepts the
word $xa^\omega$ which cannot be accepted by one of the automata
$A^{e_i}$ in the extrapolated sequence of origin $A^{e_0}$.
\end{example}

The example above shows that applying the construction of Proposition
\ref{extra-finite} to $A^{e_0}$ may introduce new cycles from states
of $Q_{I_{0}}^{e_0}$ to themselves. Since the accepting runs of the
$A^{e_i}$ can only go through a {\em finite number} of increments, it
is essential to make these cycles nonaccepting. The problem can easily
be solved, as stated with the following proposition.

\begin{proposition}
\label{extra-infinite}
Let $A^{e_0}$ defined over $\Sigma$ be a minimal weak B\"uchi
automaton which is the last element of an incrementally growing
sequence of automata $S_I$.  Assume that
$\it{GROW}_{(S_I)}(A^{e_0})={\lbrace}Q_H^{e_0},{\lbrace}Q_{I_0}^{e_0},
\ldots,Q_{I_{k-1}}^{e_0}{\rbrace},Q_{T_f}^{e_0}{\rbrace}$. One can
compute a weak B\"uchi automaton $A^{e_*}$ that represents the limit
of the extrapolated sequence of origin $A^{e_0}$.
\end{proposition}

\begin{proof}
Let $\delta$ be the transition relation of $A^{e_0}$. The automaton
$A^{e_*}$ that represents the limit of the extrapolated sequence whose
origin is $A^{e_0}$ can be built from $A^{e_0}$ by augmenting its set
of states and transitions with the following rules:

\begin{enumerate}
\item
Build an isomorphic copy $A_{I_{0}\mbox{\footnotesize\it copy}}$ of
the automaton formed by the states in $Q_{I_0}^{e_0}$, the transitions
between them, and the outgoing transitions from these states to states
in $Q_{I_1}^{e_0}$, $Q_{I_2}^{e_0}$, \ldots, $Q_{I_{k-1}}^{e_0}$, and
$Q_{T_f}^{e_0}$;
\item
Make all the states of $A_{I_{0}\mbox{\footnotesize\it copy}}$
nonaccepting;
\item
For each state $q \in Q_{I_0}^{e_0}\,\cup\,Q_H^{e_0}$ and $a \in
\Sigma$, if $\delta(q,a)$ leads to a state $q'$ in an increment
$Q_{I_{j}}^{e_0}$, $1\leq j \leq k -1$, then
\begin{enumerate}
\item
For each $1\leq \ell < j$, add a transition
$(q,a,q'')$, where $q''$ is the state corresponding to $q'$ (by the
increment isomorphism) in $Q_{I_{\ell}}^{e_0}{}$. Also, add a
transition $(q,a,q'')$, where $q''$ is the state corresponding to $q'$
in $A_{I_{0}\mbox{\footnotesize\it copy}}$;
\item
If $q \in Q_{I_0}$, then let $q_{\mbox{\footnotesize\it copy}}$ be the
state corresponding to $q$ in $A_{I_{0}\mbox{\footnotesize\it
copy}}$. For each $1\leq \ell < j$, add a transition
$(q_{\mbox{\footnotesize\it copy}},a,q'')$, where $q''$ is the state
corresponding to $q'$ (by the increment isomorphism) in
$Q_{I_{\ell}}^{e_0}{}$. Also, add a transition
$(q_{\mbox{\footnotesize\it copy}},a,q'')$, where $q''$ is the state
corresponding to $q'$ in $A_{I_{0}\mbox{\footnotesize\it copy}}$.
\end{enumerate}
\end{enumerate}
\end{proof}

\noindent
The construction in the proposition above follows from the one given
in Proposition \ref{extra-finite}. The only slight difference is in
the duplication of the head increment, which is needed to make sure
that new cycles added to $A^{e_0}$ are nonaccepting.

\begin{figure}[t]
\centering 
\label{counter-extra3}
\includegraphics[width=9cm]{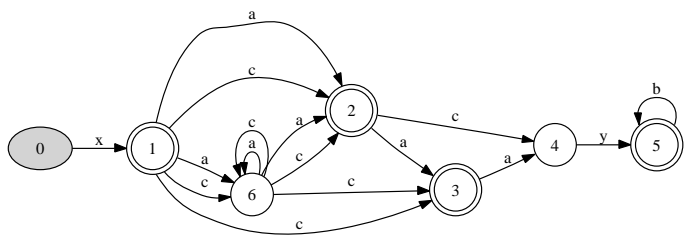}
\label{counter-extra3}
\caption{A weak B\"uchi automaton for Example \ref{example-extrapolation2} .}
\label{ex-counters3}
\end{figure}

\begin{example}
\label{example-extrapolation2}
The automaton in Figure \ref{ex-counters3} is the result of applying
the construction of Proposition \ref{extra-infinite} to Automaton
$A^{e_0}$ of Example \ref{example-extrapolation1}.
\end{example}

\begin{figure}[t]
\centering 
\subfigure[$A^{e_0}$]{
\label{counter-prop1}
\includegraphics[width=6.5cm]{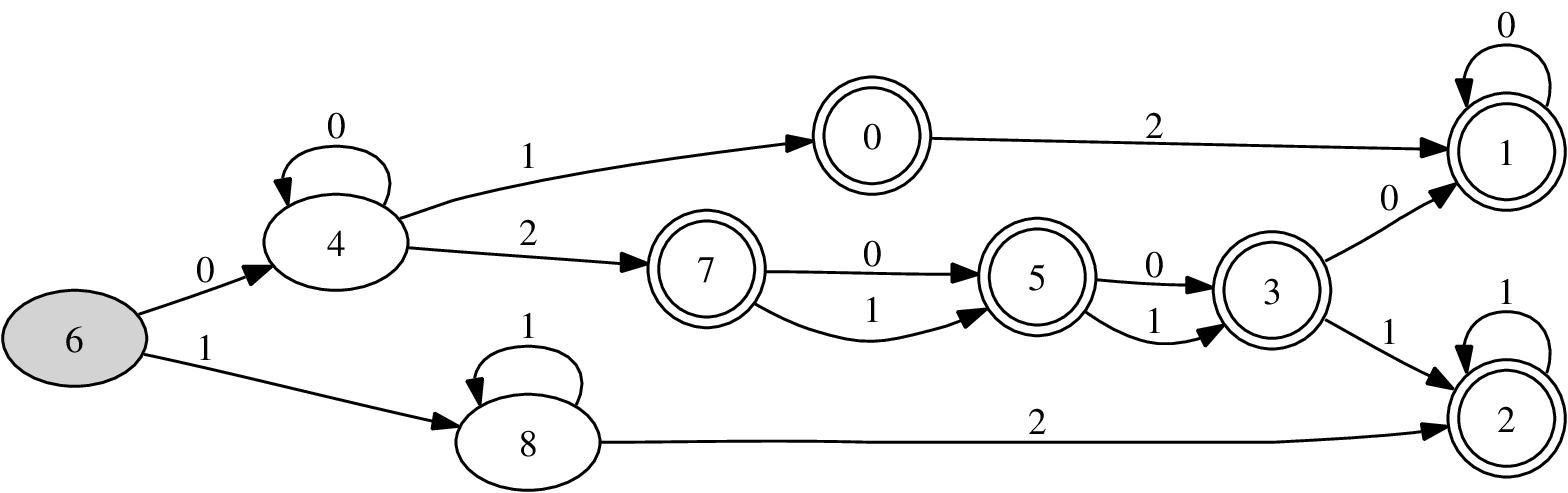}
\label{counter-prop1}
}
\hspace{2em}
\subfigure[$A^{e_*}_1$]{
\label{counter-prop2}
\includegraphics[width=6.5cm]{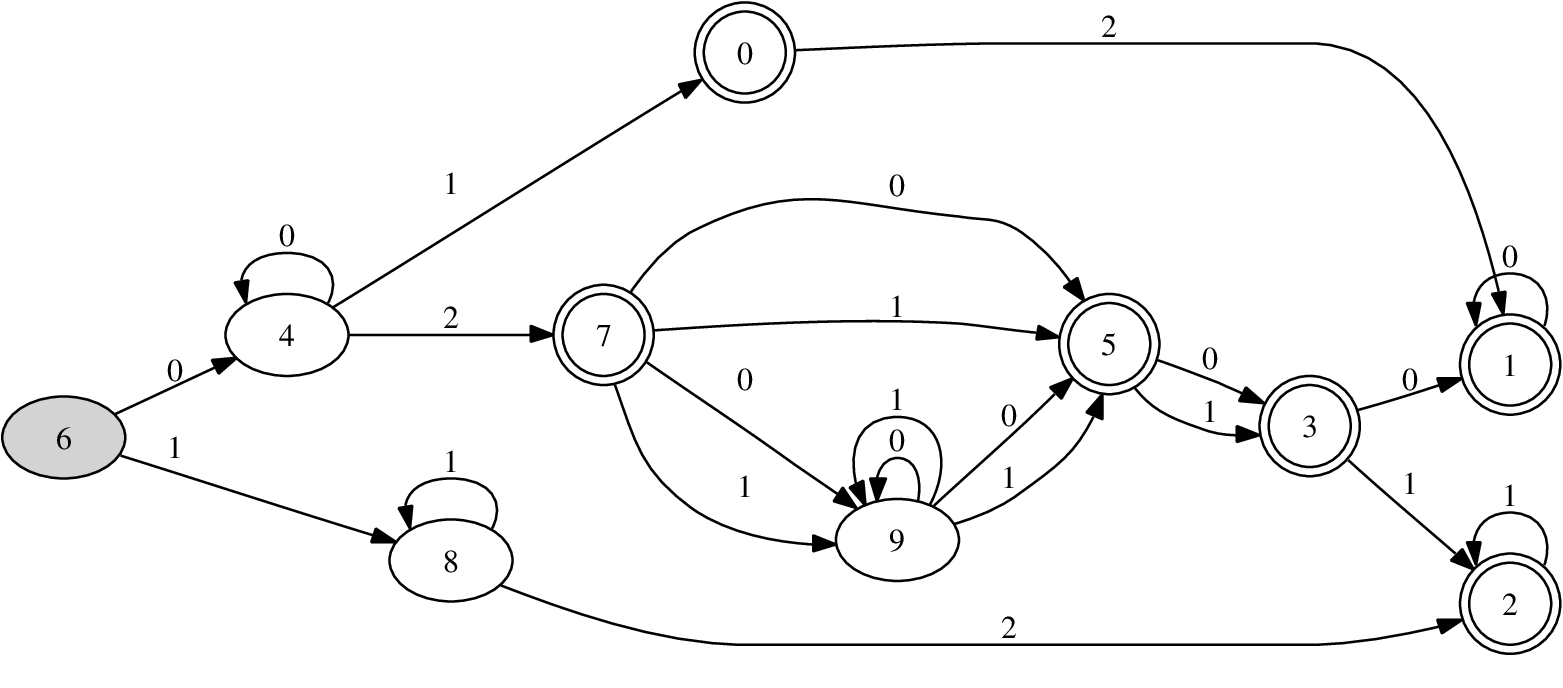}
\label{counter-prop2}
}
\hspace{2em}
\subfigure[$A^{e_*}_2$]{
\label{counter-prop3}
\includegraphics[width=8cm]{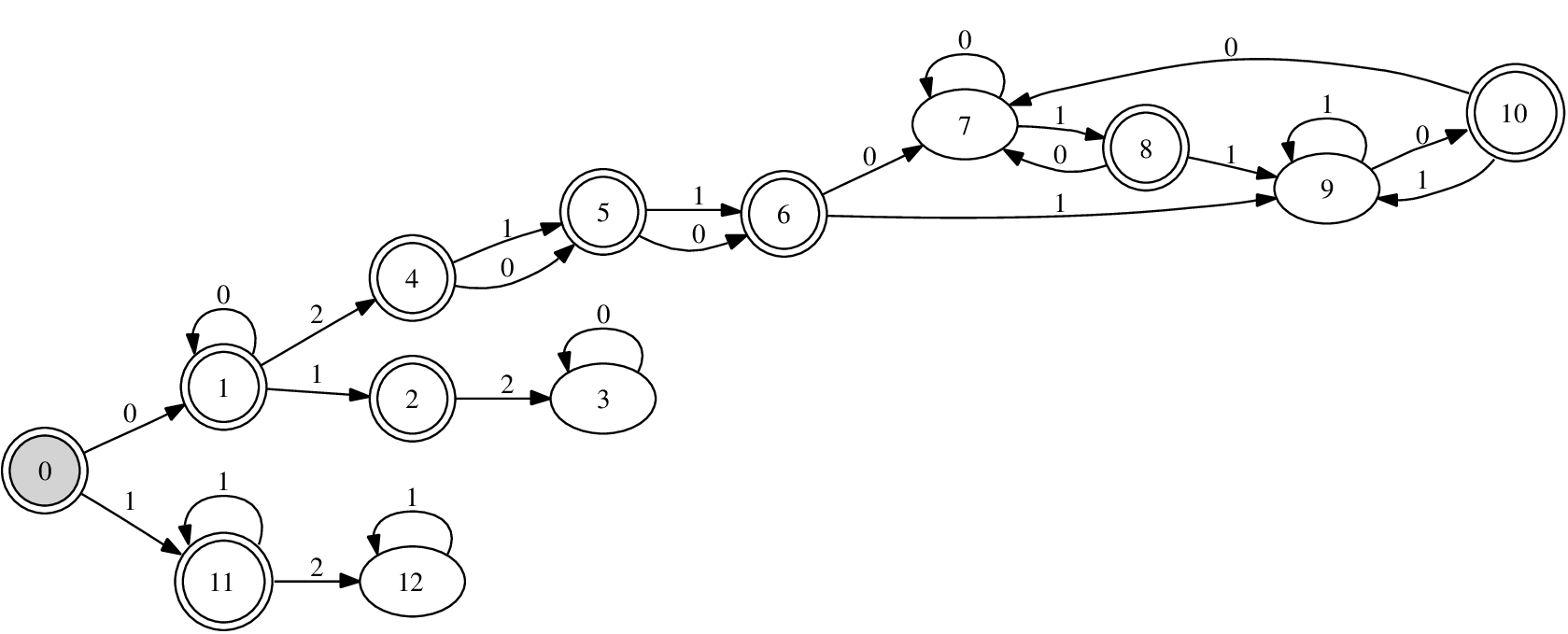}
\label{counter-prop3}
}
\caption{B\"uchi automata for the proof of Proposition
\ref{counter-not-weak} .}
\label{ex-counters2}
\end{figure}



\begin{proposition}
\label{counter-not-weak}
Let $A^{e_*}$ be the result of applying the construction of
Proposition \ref{extra-infinite} to $A^{e_0}$, the last automaton of a
finite incrementally growing sequence of deterministic weak B\"uchi
automata. The automaton $A^{e_*}$ may not be weak deterministic.
\end{proposition}

\begin{proof}
Consider the minimal weak B\"uchi automaton $A^{e_0}$ given in Figure
\ref{counter-prop1}, with $Q_H^{e_0}={\lbrace}6,4{\rbrace}$,
$Q_{I_0}^{e_0}={\lbrace}7{\rbrace}$,
$Q_{I_1}^{e_0}={\lbrace}5{\rbrace}$, and
$Q_{T_f}^{e_0}={\lbrace}0,1,2,3,8{\rbrace}$.  Applying the
construction of Proposition \ref{extra-infinite} to $A^{e_0}$ gives
the nondeterministic weak B\"uchi automaton $A^{e_*}$ in Figure
\ref{counter-prop2}. In this automaton, the state labeled by $9$ is
the duplication of $Q_{I_0}^{e_0}$. The result of determinizing
$A^{e_*}_1$ is the deterministic co-B\"uchi automaton $A^{e_*}_2$ that
is given in Figure \ref{counter-prop3}. It is easy to see that this
automaton is not inherently weak and, consequently, cannot be turned
to a weak B\"uchi automaton.
\end{proof}



Following what has been done for the case of finite-word automata, we
now propose to add a counter $c$ to $A^{e_*}$ in such a way that
when a word is accepted, the value of $c$ is the smallest index $i$ of
the automaton $A^{e_i}$ of the extrapolated sequence by which the
word is in fact accepted.

\begin{proposition}
\label{extra-infinite-counter}
Let $A^{e_0}=(Q,\Sigma,Q_0,\delta,F)$ be a minimal weak B\"uchi
automaton which is the last element of an incrementally growing
sequence of automata $S_I$.  Assume that
$\it{GROW}_{(S_I)}(A^{e_0})={\lbrace}Q_H^{e_0},{\lbrace}Q_{I_0}^{e_0},
\ldots,Q_{I_{k-1}}^{e_0}{\rbrace},Q_{T_f}^{e_0}{\rbrace}$ and let
$A^{e_0}, A^{e_1},\dots$ be the extrapolated sequence of origin
$A^{e_0}$. One can compute a run-bounded weak B\"uchi counter
automaton $A^{e_*}_c$ such that (1)
$L(A^{e_*}_c)={\bigcup}_{i{\geq}0}A^{e_i}$, (2) for each $(w,i)$ $\in$
${\cal{L}}(A^{e_*}_c)$, $w\in L(A^{e_i})$, and (3) for each $w$ $\in$
$L(A^{e_i})$, $j{\leq}i$ exists such that $(w,j)$ $\in$
${\cal{L}}(A^{e_*}_c)$.
\end{proposition}

\begin{proof}
Let $\delta$ be the transition relation of $A^{e_0}$. The
one-dimensional counter automaton $A^{e_*}_c$ is given by
$(1,\vect{c},Q',\Sigma,Q_0,\triangle,F)$ , with $Q$ and $\triangle$
defined as follows:
\begin{enumerate}
\item
Start with $\triangle={\lbrace}\emptyset{\rbrace}$;
\item
For each $(q,a,q')\in \delta$, add $(q,(a,\vect{0}),q')$ to
$\triangle$;
\item
Build an isomorphic copy $A_{I_{0}\mbox{\footnotesize\it copy}}$ of
the automaton formed by the states in $Q_{I_0}^{e_0}$, the transitions
between them, and the outgoing transitions from these states to states
in $Q_{I_1}^{e_0}$, $Q_{I_2}^{e_0}$, \ldots, $Q_{I_{k-1}}^{e_0}$, and
$Q_{T_f}^{e_0}$. All the transitions are associated with the counter
increment $0$;
\item
Make all the states of $A_{I_{0}\mbox{\footnotesize\it copy}}$
nonaccepting;
\item
For each state $q \in Q_{I_0}^{e_0}\,\cup\,Q_H^{e_0}$ and $a \in
\Sigma$, if $\delta(q,a)$ leads to a state $q'$ in an increment
$Q_{I_{j}}^{e_0}$, $1\leq j \leq k -1$, then
\begin{enumerate}
\item
For each $1\leq \ell < j$, add to $\triangle$ a transition
$(q,(a,j-l),q'')$, where $q''$ is the state corresponding to $q'$ (by
the increment isomorphism) in $Q_{I_{\ell}}^{e_0}{}$. Also, add a
transition $(q,(a,j),q'')$, where $q''$ is the state corresponding to
$q'$ in $A_{I_{0}\mbox{\footnotesize\it copy}}$;
\item
If $q \in Q_{I_0}$, then let $q_{\mbox{\footnotesize\it copy}}$ be the
state corresponding to $q$ in $A_{I_{0}\mbox{\footnotesize\it
copy}}$. For each $1\leq \ell < j$, add to $\triangle$ a transition
$(q_{\mbox{\footnotesize\it copy}},(a,j-l),q'')$, where $q''$ is the
state corresponding to $q'$ (by the increment isomorphism) in
$Q_{I_{\ell}}^{e_0}{}$. Also, add a transition
$(q_{\mbox{\footnotesize\it copy}},(a,j),q'')$, where $q''$ is the
state corresponding to $q'$ in $A_{I_{0}\mbox{\footnotesize\it
copy}}$.
\end{enumerate}
\end{enumerate}
\end{proof}

\noindent
Let $A^{e_0}_c$ be the counter-zero automaton corresponding to
$A^{e_0}$. From the observations above, we directly see that for each
$i\in \nats_0$ $w$ $\in$ $L(A^{e_i})\setminus L(A^{e_0})$,
$1{\leq}j{\leq}i$ exists such that $(w,j)$ $\in$
${\cal{L}}(A^{e_*}_c)\setminus {\cal{L}}(A^{e_0}_c)$.



\begin{figure}[t]
\centering 
\includegraphics[width=10cm]{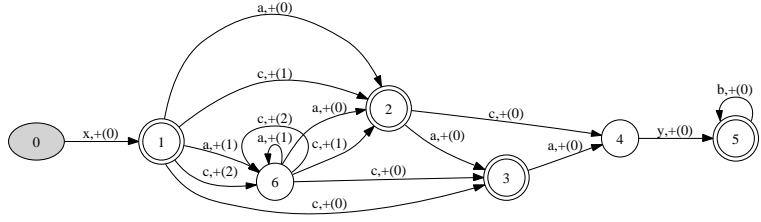}
\label{ex-counter0}
\caption{Automaton for Example \ref{example-extra3}.}
\label{extra7}
\end{figure}

\begin{example}
\label{example-extra3}
Figure \ref{extra7} presents the result of applying the construction
of Proposition \ref{extra-infinite-counter} to Automaton $A^{e_0}$ of
Example \ref{example-extrapolation1}.
\end{example}

\section{Safety and Preciseness}
\label{section-precision}

After having constructed a finite automaton $A^{e_*}$ representing the
extrapolation of a sequence $A^1$, $A^2, \ldots$ of automata, it
remains to check whether it accurately corresponds to what we really
intend to compute, i.e., $\bigcup_{i > 0} A^i$. This is done by first
checking that the extrapolation is {\em safe\/}, in the sense that it
captures all behaviors of $\bigcup_{i > 0} A^i$, and then checking
that it is {\em precise\/}, i.e., that it has no more behaviors than
$\bigcup_{i > 0} A^i$. We check both properties using sufficient
conditions. We develop separately these conditions for the two
($\omega$-)Regular Reachability Problems.

\begin{remark}
As we already mentioned in the introduction, the ability to
extrapolate an infinite sequence of automata has other applications
than solving the ($\omega$-)Regular Reachability Problems (see
\cite{BLW04b,CLW08} for examples). Depending on the problem being
considered, we may have to use other correctness criteria than those
that are proposed in this paper.
\end{remark}

\subsection{Transitive Closure of a Transducer}
\label{subsec-transd}

Consider a reflexive deterministic finite-word (respectively,
deterministic weak B\"uchi) transducer $T$ and let $T^{e_0}$ be the
last element of an incrementally growing sampling sequence $S_I$ of
powers of $T$. Assume that $T^{e_0}$ is the origin of an extrapolated
sequence $T^{e_0},T^{e_1},\dots$. The limit of this sequence is the
transducer $T^{e_*}$ with
$L(T^{e_*})=\bigcup_{i=0}^{\infty}L(T^{e_i})$ that has been computed
by applying the construction of Proposition \ref{extra-finite}
(respectively, Proposition \ref{extra-infinite}) to $T^{e_0}$. We
provide sufficient criteria to test whether $L(T^*)=L(T^{e_*})$.\\
\newline
We first determine whether $T^{e_*}$ is a safe extrapolation of $T$,
i.e., whether $L(T^*)\,\subseteq\,L(T^{e_*})$. For this, we propose
the following result.

\begin{proposition}
\label{safe-extrapolation-transdu}
Let $T_1$ and $T_2$ be two reflexive transducers defined over the same
alphabet. If $L(T_2\circ T_2)\,\subseteq\, L(T_2)$ and
$L(T_1)\,\subseteq\,L(T_2)$, then $L(T_1^*)\subseteq\,L(T_2)$.
\end{proposition}

\begin{proof}
We show by induction that for each $i>0$,
$L(T^i_1)\,\subseteq\,L(T_2)$. The base cases, i.e.,
$L(T_1^0)\,\subseteq\,L(T_2)$ and $L(T_1)\,\subseteq\,L(T_2)$, hold
by hypothesis.  Suppose now that $i>1$ and that the result holds for
any $k<i$. It is easy to see that
$L(T_1^i)\,\subseteq\,L(T_2)$. Indeed, $L(T^{i}_1)=L(T^{i-1}_1\circ
T_1)\,\subseteq\,L(T_2 \circ T_1)\,\subseteq\,L(T_2\circ
T_2)\,\subseteq\,L(T_2)$. The first inclusion holds by induction, the
second because $L(T_1)\,\subseteq\, L(T_2)$, and the third is by
hypothesis.
\end{proof}

\noindent
By construction, $L(T)\,\subseteq\,L(T^{e_*})$ and, moreover, $T$ is
reflexive. Consequently, Proposition \ref{safe-extrapolation-transdu}
states that if $L(T^{e_*}\circ T^{e_*})\,\subseteq\,L(T^{e_*})$, then
$T^{e_*}$ is a safe extrapolation of $T^*$. This criterion is only
sufficient since their could exist two words $w,w'\in L(T^{e_*})$ such
that $w,w'\not\in L(T^*)$ and $w~{\circ}~w'$ $\not\in$
$L(T^{e_*})$. In practice, checking the condition expressed by
Proposition \ref{safe-extrapolation-transdu} requires to complement
$T^{e_*}$. Indeed, this condition is equivalent to checking whether
the language accepted by the automaton which is the intersection of
the automaton for $T^{e_*}\circ T^{e_*}$ and the one for the
complement of $T^{e_*}$ is empty or not.  When working with weak
automata, $T^{e_*}$ is by construction weak but generally not
deterministic (see Proposition \ref{counter-not-weak}). Our approach
consists in determinizing $T^{e_*}$, and then checking whether the
resulting transducer is inherently weak. In the positive case, this
transducer can be turned into a weak deterministic one and easily be
complemented by inverting the sets of accepting and nonaccepting
states. Otherwise a B\"uchi complementation algorithm has to be
applied.\\
\newline 
We now turn to determine whether $T^{e_*}$ is a precise extrapolation
of $T$, i.e., whether $L(T^{e_*})\,\subseteq\,L(T^{*})$. For this, we
again provide a partial solution in the form of a sufficient
criterion. The ``preciseness'' problem amounts to proving that any
word accepted by $T^{e_*}$, or equivalently by some $T^{e_i}$, is also
accepted by an iteration $T^j$ of the transducer $T$. The idea is to
check that this can be proved inductively. The property is true by
construction for the transducer $T^{e_0}$ from which the extrapolation
sequence is built. If we can also prove that, if the property holds
for all $T^{e_j}$ with $j<i$, then it also holds for $T^{e_i}$, we are
done. For this, we propose the following theorem.

\begin{theorem}
\label{precis-extrapolation-transdu1}
Let $T$ and $T^{e_*}$ be two transducers and $T^{e_0}$ be a power of
$T$. Assume an infinite sequence of transducers
$T^{e_0},T^{e_1},\dots$, and let
$L(T^{e_*})=\bigcup_{i=0}^{\infty}L(T^{e_i})$. If
\begin{equation}
\label{cond-extra-transdu1}
 \forall w, {\forall} i >0 \ {\lbrack}w \in L(T^{e_i})\setminus
L(T^{e_0}) \ {\Rightarrow} \ {\exists}0\,{\leq}\,j,j'<i, w \in
L(T^{e_j}{\circ}T^{e_j'}){\rbrack},
\end{equation}
then $L(T^{e_*})\,\subseteq\,L(T^*)$.
\end{theorem}

\begin{proof}
The proof is by induction: we show that for each $i\,{\geq}\,0$,
$L(T^{e_i})\subseteq L(T^*)$. The base case, i.e.,
$L(T^{e_0})\,\subseteq\,L(T^*)$, holds by hypothesis. Suppose now that
$i>0$ and that the result holds for any $j<i$. We show that
$L(T^{e_i})\,\subseteq\,L(T^*)$. Consider a word $w\in L(T^{e_i})$. If
$w\in L(T^{e_0})$, then the result holds. If $w\notin L(T^{e_0})$
then, by Condition (\ref{cond-extra-transdu1}) there exist $j,j'<i$,
$w'\in L(T^{e_j})$, $w''\in L(T^{e_{j'}})$ such that $w=w'\circ
w''$. Since, by inductive hypothesis $w',w''\in L(T^*)$, $n_1, n_2\in
\nats$ exist such that $w'\in L(T^{n_1})$ and $w''\in L(T^{n_2})$. We
thus have $w\in L(T^{n_1+n_2})$.
\end{proof}

\noindent
Theorem \ref{precis-extrapolation-transdu1} reduces the problem of
checking the preciseness of $T^{e_*}$ to the one of testing whether
Condition (\ref{cond-extra-transdu1}) is satisfied or not. We now go
one step further and reduce this test to automata-based manipulations.

\begin{lemma}
\label{precis-extrapolation-transdu2}
Let $T^{e_0}$ be the last element of an incrementally growing sampling
sequence $S_I$ of transducers, and $T^{e_0}_c$ be the counter-zero
automaton corresponding to $T^{e_0}$. Assume that $T^{e_0}$ is the
origin of an extrapolated sequence $T^{e_0},T^{e_1},\dots$ and let
$T^{e_*}_{c_1}, T^{e_*}_{c_2}, T^{e_*}_{c_3}$ be three copies of the
counter transducer $T^{e_*}_{c}$ which is obtained by applying the
construction of Proposition \ref{extra-finite-counter} (respectively,
Proposition \ref{extra-infinite-counter}) to $T^{e_0}$. If
\begin{equation}
\label{cond-extra-transdu2}
{\cal{L}}({\pi}_{(\not={\lbrace}c_2,c_3{\rbrace})}{\lbrack}(T_{c_1}^{e_*}
\cap_c (T_{c_2}^{e_*} \circ_c
T_{c_3}^{e_{*}}))^{c_1>{\lbrace}c_2,c_3{\rbrace}}){\rbrack})=
{\cal{L}}(T_{c}^{e_*})\setminus{\cal{L}}(T_{c}^{e_0}),
\end{equation} 
\noindent
then
\begin{equation*}
 \forall w, {\forall} i >0 \ {\lbrack}w \in L(T^{e_i})\setminus
L(T^{e_0}) \ {\Rightarrow} \ {\exists}0\,{\leq}\,j,j'<i, w \in
L(T^{e_j}{\circ}T^{e_j'}){\rbrack}.
\end{equation*}
\end{lemma}

\begin{proof}
Observe that the counter language of
${\pi}_{(\not={\lbrace}c_2,c_3{\rbrace})}{\lbrack}(T_{c_1}^{e_*}
\cap_c (T_{c_2}^{e_*} \circ_c
T_{c_3}^{e_{*}}))^{c_1>{\lbrace}c_2,c_3{\rbrace}}$ is the counter
language of $T_{c}^{e_*}$ from where one has removed all the pairs
$(w,i)$ for which there is not $(w',j<i), (w'',j'<i) \in
{\cal{L}}(T_{c}^{e_*})$ with $w=w'\circ w''$. For each $i$ and each
word $w$, if $w\in L(T^{e_i})\setminus L(T^{e_0})$ then, by
Proposition \ref{extra-finite-counter} (respectively, Proposition
\ref{extra-infinite-counter}), there exists $k>0\in \nats$ such that
$(w,k{\leq}i)\in {\cal{L}}(T_c^{e_*})\setminus
    {\cal{L}}(T_c^{e_0})$. Since Condition (\ref{cond-extra-transdu2})
    holds, there exist $j,j'\in \nats$ with $j,j'\,<\,k\,{\leq}\,i$
    and two words $w',w''$ such that $(w',j)\in {\cal{L}}(T_c^{e_*})$
    and $(w'',j')\in {\cal{L}}(T_c^{e_*})$, with $w=w'\circ w''$. By
    Proposition \ref{extra-finite-counter} (respectively, Proposition
    \ref{extra-infinite-counter}), $w'\in L(T^{e_{j}})$ and $w''\in
    L(T^{e_{j'}})$ and $w\in L(T^{e_{j}}\circ T^{e_{j'}})$.
\end{proof}


We can now state our main result.

\begin{theorem}
\label{main-result-precis-transdu}
Let $T$ be a transducer, $T^{e_0}$ the last element of an
incrementally growing sampling sequence $S_I$ of powers of $T$, and
$T^{e_0}_c$ the counter-zero automaton corresponding to $T^{e_0}$.
Assume that $T^{e_0}$ is the origin of an extrapolated sequence
$T^{e_0},T^{e_1},\dots$ and let $T^{e_*}$ be the transducer that has
been obtained by applying the construction of proposition
\ref{extra-finite} (respectively, Proposition \ref{extra-infinite}) to
$T^{e_0}$. Let $T^{e_*}_{c_1}, T^{e_*}_{c_2}, T^{e_*}_{c_3}$ be three
copies of the counter transducer $T^{e_*}_c$ which is obtained by
applying the construction of Proposition \ref{extra-finite-counter}
(respectively, Proposition \ref{extra-infinite-counter}) to
$T^{e_0}$. If
${\cal{L}}({\pi}_{(\not={\lbrace}c_2,c_3{\rbrace})}{\lbrack}(T_{c_1}^{e_*}
\cap_c (T_{c_2}^{e_*} \circ_c
T_{c_3}^{e_{*}}))^{c_1>{\lbrace}c_2,c_3{\rbrace}}){\rbrack})=
{\cal{L}}(T_{c_1}^{e_*})\setminus{\cal{L}}(T_{c}^{e_0})$, then
$L(T^{e_*})\subseteq L(T^*)$.
\end{theorem}

\begin{proof}
By Proposition \ref{extra-finite} (respectively, Proposition
\ref{extra-infinite}), we have
$L(T^{e_*})=\bigcup_{i=0}^{\infty}L(T^{e_i})$.\\
\newline
\noindent
According to Lemma \ref{precis-extrapolation-transdu2}, since
\begin{equation*}
{\cal{L}}({\pi}_{(\not= {\lbrace} c_2,c_3{\rbrace})}{\lbrack}(T_{c_1}^{e_*}
\cap_c (T_{c_2}^{e_*} \circ_c
T_{c_3}^{e_{*}}))^{c_1>{\lbrace}c_2,c_3{\rbrace}}){\rbrack})=
{\cal{L}}(T_{c}^{e_*})\setminus{\cal{L}}(T_{c}^{e_0}),
\end{equation*} 
\noindent
we have
\begin{equation*}
 \forall w, {\forall} i >0 \ {\lbrack}w \in L(T^{e_i})\setminus
L(T^{e_0}) \ {\Rightarrow} \ {\exists}0\,{\leq}\,j,j'<i, w \in
L(T^{e_j}{\circ}T^{e_j'}){\rbrack}.
\end{equation*}

\noindent
It follows from Theorem \ref{precis-extrapolation-transdu1} that
$L(T^{e_*})\,\subseteq\,L(T^*)$.
\end{proof}

\noindent
Condition (\ref{cond-extra-transdu2}) can be implemented as follows :

\begin{itemize}
\item
Observe that, since ${\cal{L}}({\pi}_{(\not= {\lbrace}
  c_2,c_3{\rbrace})}{\lbrack}(T_{c_1}^{e_*} \cap_c (T_{c_2}^{e_*}
\circ_c T_{c_3}^{e_{*}}))^{c_1>{\lbrace}c_2,c_3{\rbrace}}){\rbrack})$
is disjoint from, checking ${\cal{L}}({\pi}_{(\not= {\lbrace}
  c_2,c_3{\rbrace})}{\lbrack}(T_{c_1}^{e_*} \cap_c (T_{c_2}^{e_*}
\circ_c T_{c_3}^{e_{*}}))^{c_1>{\lbrace}c_2,c_3{\rbrace}}){\rbrack})=
     {\cal{L}}(T_{c}^{e_*})\setminus{\cal{L}}(T_{c}^{e_0})$ is
     equivalent to check ${\cal{L}}({\pi}_{(\not= {\lbrace}
       c_2,c_3{\rbrace})}{\lbrack}(T_{c_1}^{e_*} \cap_c (T_{c_2}^{e_*}
     \circ_c
     T_{c_3}^{e_{*}}))^{c_1>{\lbrace}c_2,c_3{\rbrace}}){\rbrack}\cup_e
     T_{c}^{e_0})= {\cal{L}}(T_{c}^{e_*})$, which avoid to compute
     ${\cal{L}}(T_{c}^{e_*})\setminus{\cal{L}}(T_{c}^{e_0})$. Computing
     ${\cal{L}}(T_{c}^{e_*})\setminus{\cal{L}}(T_{c}^{e_0})$ is a hard
     problem, which requires the ability to distinguish between
     accepting and nonaccepting runs that assign the same counter
     valuation to a given word.
\item
There are algorithms to compute $\cap_c$, $\circ_c$, and
${\pi}_{(\not={\lbrace}c_2,c_3{\rbrace})}$. Those algorithms directly
follow from the definitions given in Section
\ref{counter-section}. Observe that if $T$ is weak, then the counter
automaton for $T_{c_1}^{e_*} \cap_c (T_{c_2}^{e_*} \circ_c
T_{c_3}^{e_{*}})$ is run-bounded weak.
\item
We do not compute the one-counter automaton for $(T_{c_1}^{e_*} \cap_c
(T_{c_2}^{e_*} \circ_c
T_{c_3}^{e_{*}}))^{c_1>{\lbrace}c_2,c_3{\rbrace}}$, but a
$M$-synchronized counter automaton whose language and counter
languages may be subsets of those of $(T_{c_1}^{e_*} \cap_c
(T_{c_2}^{e_*} \circ_c
T_{c_3}^{e_{*}}))^{c_1>{\lbrace}c_2,c_3{\rbrace}}$. We follow the
methodology described in Section \ref{counter-section}, and compute
the extended-intersection between the automaton $T_{c_1}^{e_*} \cap_c
(T_{c_2}^{e_*} \circ_c T_{c_3}^{e_{*}})$ and two finite-word
(respectively, run-bounded weak B\"uchi) $M$-Universal-synchronized
counter automata, one which is synchronized w.r.t. counters $c_1$ and
$c_2$, and the other one w.r.t. counters $c_1$ and $c_3$. Assume that
$\Sigma^2$ is the alphabet of $T$ and $d$ is the maximal increment
value of $T^{e_*}_c$. The extended alphabet of $T_c$ is $\Sigma\times
{\lbrack}0,d{\rbrack}$, and the one of $T_{c_1}^{e_*} \cap_c
(T_{c_2}^{e_*} \circ_c T_{c_3}^{e_{*}})$ is thus $\Sigma^2\times
{\lbrack}0,d{\rbrack}^3$ (see constructions for $\circ_c$ and
$\cap_c$). In our experiments (see \cite{Leg07} for details), we
worked with counter automata whose extended alphabet is
$\Sigma^2\times {\lbrack}0,d{\rbrack}^3$, and such that $c_1$ is
$M$-synchronized with respect to $c_2$ and $c_3$, with
$M=2{\times}d$. This choice turned out to be the best compromise for
our experimental results\,\cite{Leg07,T(O)RMC}, where we clearly
observed a synchronization between the counters.
\item
We reduce the problem of checking the equivalence between the counter
languages of the two members of the equality to the one of checking
the equivalence between the languages of their extended automata (see
Proposition \ref{th-counter1}).
\end{itemize}

\noindent
Observe that, if $L(T^*)=L(T^{e_*})$, then the transducers $T^{e_i}$
($i{\geq}0$) may constitute new elements in an extension of the
sampling sequence $S_i$, i.e., if $S_I=T^{s_0},T^{s_1},\dots,T^{s_k}$
with $T^{s_k}=T^{e_0}$, then the extension is
$T^{s_0},T^{s_1},\dots,T^{s_k}, T^{s_{k+1}}, $\\$T^{s_{k+2}},\dots$,
with $T^{s_{k+i}}=T^{e_i}$ for each $i{\geq}0$. Condition
(\ref{cond-extra-transdu1}) is thus particularly designed to hold for
sampling sequences where each transducer can be obtained by a single
composition of transducers that appear before in the sequence. Indeed,
the condition can be read as follows: {\em each transducer $T^{e_i}$
  in the extended sampling sequence is the composition of two
  transducers $T^{e_j}$ and $T^{e_{j'}}$ that appear before in this
  sequence}. If more than one composition is needed, then the
condition may not be satisfied even if $L(T^{e_*})=L(T^*)$. Condition
(\ref{cond-extra-transdu1}) can be adapted to work with other sampling
sequences. This is illustrated with the following example.

\begin{example}
If each transducer in the sampling sequence is obtained by composing
$n$ transducers that appear before in the sequence, then one can test
whether the following condition holds

\begin{equation}
 \forall w, {\forall} i >0 \
{\lbrack}w \in L(T^{e_i})\setminus L(T^{e_0}) \ {\Rightarrow} \ {\exists}0\,{\leq}j_1,\dots,j_n<i, w \in
L(T^{e_{j_1}}{\circ}\dots {\circ}T^{e_{j_n}}){\rbrack},
\end{equation} 

\noindent
rather than to test whether Condition (\ref{cond-extra-transdu1})
holds.
\end{example}

\noindent
Theorem \ref{main-result-precis-transdu} easily extends to other
sampling sequences.

\subsection{Limit of a Sequence of Reachable Sets}
\label{subsec-autopreci}

This section lifts the results obtained in the previous section to the
case where one computes the limit of a sequence of reachable states.
We consider a reflexive finite-word (respectively, deterministic weak
B\"uchi) transducer $T$ and a deterministic finite-word (respectively,
deterministic weak B\"uchi) automaton $A$. Let $A^{e_0}$ be the last
automaton of an incrementally growing sampling sequence $S_I$ of $A$,
$T^{1}(A)$, $T^{2}(A)$, $T^{3}(A)$, and assume that $A^{e_0}$ is the
origin of an extrapolated sequence $A^{e_0},A^{e_1},\dots$ . The limit
of this sequence is the automaton $A^{e_*}$ with
$L(A^{e_*})=\bigcup_{i=0}^{\infty}L(A^{e_i})$ that has been computed
by applying the construction of Proposition \ref{extra-finite}
(respectively, Proposition \ref{extra-infinite}) to $A^{e_0}$. We
provide sufficient criteria to test whether $L(T^*(A))=L(A^{e_*})$.\\
\newline
We first determine whether $A^{e_*}$ is a safe extrapolation of
$T^*(A)$, i.e., whether
$L(T^*(A))\,\subseteq\,L(A^{e_*})$. For this, we propose the following
result.

\begin{proposition}
\label{safe-extrapolation-auto}
Let $A_1$ and $A_2$ be two automata defined over the same alphabet
$\Sigma$ and with $L(A_1)\,\subseteq\,L(A_2)$. Let $T$ be a reflexive
transducer over $\Sigma^2$. If $L(T(A_2))\,\subseteq\,L(A_2)$ then
$L(T^*(A_1))\,\subseteq\,L(A_2)$.
\end{proposition}

\begin{proof}
By hypothesis, we have $L(A_1)\,\subseteq\, L(A_2)$.  We show by
induction that for each $i>0$, $L(T^i(A_1))\,\subseteq\,L(A_2)$. The
base cases, i.e., $L(A_1)\,\subseteq\, L(A_2)$ and
$L(T(A_1))\,\subseteq\,L(A_2)$, hold by hypothesis.  Suppose now that
$i>1$ and that the result holds for any $j<i$. It is easy to see that
$L(T^i(A_1))\,\subseteq\,L(A_2)$. Indeed,
$L(T^{i}(A_1))=L(T(T^{i-1}(A_1)))\,\subseteq\,L(T(A_2))\,\subseteq\,L(A_2)$.
The first inclusion holds by induction and the second because
$L(T(A_2))\,\subseteq\,L(A_2)$.
\end{proof}

\noindent
Proposition \ref{safe-extrapolation-auto} states that checking whether
$A^{e_*}$ is a safe extrapolation of $\bigcup_{i=0}^{\infty}T^i(A)$
can be done by checking whether
$L(T(A^{e_*}))\,\subseteq\,L(A^{e_*})$.  It is worth mentioning that
this criterion is only sufficient. Indeed, their could exist a word
$w\in L(A^{e_*})$ such that $w \not\in L(T^*(A))$ and $w\not\in$
$L(T(A^{e_*}))$.\\
\newline
We now turn to determine whether $A^{e_*}$ is a precise extrapolation
of $T^*(A)$, i.e., whether $L(A^{e_*})\,\subseteq\,L(T^{*}(A))$. As in
Section \ref{subsec-transd}, we use an inductive argument, which is
formalized with the following theorem.

\begin{theorem}
\label{precis-extrapolation-auto1}
Let $T$ be a transducer and $A, A^{e_*}$ be two automata. Let
$A^{e_0}=T^k(A)$, and consider an infinite sequence of automata
$A^{e_0},A^{e_1},\dots$, with
$L(A^{e_*})=\bigcup_{i=0}^{\infty}L(A^i)$. If
\begin{equation}
\label{cond-extra-auto1}
 \forall w, {\forall} i >0 \ {\lbrack}w \in L(A^{e_i})\setminus
L(A^{e_0}) \ {\Rightarrow} \ {\exists}0\,{\leq}\,j<i, w \in
L(T(A^{e_j})){\rbrack},
\end{equation}
then $L(A^{e_*}))\,\subseteq L(T^*(A)$.
\end{theorem}

\begin{proof}
The proof is by induction: we show that for each $i\,{\geq}\,0$,
$L(A^{e_i})\subseteq L(T^*(A))$. The base case, i.e.,
$L(A^{e_0})\,\subseteq\,L(T^*(A))$, holds by hypothesis. Suppose now
that $i>0$ and that the result holds for any $j<i$. We show that
$L(A^{e_i})\,\subseteq\,L(T^*)$. Consider a word $w\in L(A^{e_i})$. If
$w\in L(A^{e_0})$, then the result holds. Assume now that $w\notin
L(A^{e_0})$. By Condition (\ref{cond-extra-auto1}), there exists $j<i$
such that $w\in L(T(A^{e_j}))$. Since, $T$ is reflexive and by
inductive hypothesis, there exists $n$ such that
$L(A^{e_j})\,\subseteq\,L(T^n(A))$. We thus have $w\in L(T^{n+1}(A))$.
\end{proof}

We now go one step further and reduce the verification of Condition
(\ref{cond-extra-auto1}) to simple automata-based manipulations.

\begin{lemma}
\label{precis-extrapolation-auto2}
Let $T$ be a reflexive transducer and $A$ be an automaton. Let
$A^{e_0}$ be the last automaton of an incrementally growing sampling
sequence $S_I$ of $A$, $T^{1}(A)$, $T^{2}(A)$, $T^{3}(A)$, and assume
that $A^{e_0}$ is the origin of an extrapolated sequence
$A^{e_0},A^{e_1},\dots$ and let $A_{c_1}^{e_*}, A_{c_2}^{e_*}$ be two
copies of the counter automaton $A_c^{e_*}$ that is obtained by
applying the construction of Proposition \ref{extra-finite-counter}
(respectively, Proposition \ref{extra-infinite-counter}) to
($A^{e_0}$,$\it{GROW}_{(S_I)}(A^{e_0})$). Let $A^{e_0}_c$ be the
counter-zero automaton corresponding to $A^{e_0}$. If
\begin{equation}
\label{cond-extra-auto2}
{\cal{L}}({\pi}_{(\not= c_2)}{\lbrack}(A_{c_1}^{e_*} \cap_c
T(A^{e_*}_{c_2}))^{c_1>c_2}){\rbrack})=
{\cal{L}}(A_{c}^{e_*})\setminus {\cal{L}}(A_{c}^{e_0}),
\end{equation} 
\noindent
then
\begin{equation*}
\forall w, {\forall} i >0 \ {\lbrack}w \in L(A^{e_i})\setminus
L(A^{e_0}) \ {\Rightarrow} \ {\exists}0\,{\leq}\,j<i, w \in
L(T(A^{e_j})){\rbrack}.
\end{equation*}
\end{lemma}

\begin{proof}
Observe that the counter language of ${\pi}_{(\not=
  c_2)}{\lbrack}(A_{c_1}^{e_*} \cap_c T(A^{e_*}_{c_2}))^{c_1>c_2}$ is
the counter language of $A_{c}^{e_*}$ from where one has removed all
the pairs $(w,i)$ for which there is no pair $(w',j<i)\in
{\cal{L}}(A_{c}^{e_*})$ with $w\in L(T(A_{w'}))$ (where $A_{w'}$ is an
automaton whose language is ${\lbrace}w'{\rbrace}$) have been
removed. For each $i$ and each word $w$, if $w\in L(A^{e_i})\setminus
L(A^{e_0})$ then, by Proposition \ref{extra-finite-counter}
(respectively, Proposition \ref{extra-infinite-counter}), there exists
$k>0\in \nats$ such that $(w,k{\leq}i)\in {\cal{L}}(A_c^{e_*})$. Since
Condition (\ref{cond-extra-auto2}) holds, there exists $j\in \nats$
with $j<k\,{\leq}\,i\in \nats$ and a word $w'$ such that $(w',j)\in
{\cal{L}}(T_c^{e_*})$ with $w=L(T(A_{w'}))$. By Proposition
\ref{extra-finite-counter} (respectively, Proposition
\ref{extra-infinite-counter}), $w'\in L(A^{e_{j}})$ and $w\in
L(T(A^{e_{j}}))$.
\end{proof}

Finally, we obtain our main result.

\begin{theorem}
\label{main-result-precis-auto}
Let $T$ be a reflexive transducer and $A$ be an automaton. Let
$A^{e_0}$ be the last automaton of an incrementally growing sampling
sequence $S_I$ of $A$, $T^{1}(A)$, $T^{2}(A)$, $T^{3}(A)$, and assume
that $A^{e_0}$ is the origin of an extrapolated sequence
$A^{e_0},A^{e_1},\dots$ Let $A^{e_*}$ be the automaton that has been
obtained by applying the construction of Proposition
\ref{extra-finite} (respectively, Proposition \ref{extra-infinite}) to
$A^{e_0}$, and let $A^{e_*}_{c_1}, A^{e_*}_{c_2}$ be two copies of the
counter automaton $A^{e_*}_c$ that is obtained by applying the
construction of Proposition \ref{extra-finite-counter} (respectively,
Proposition \ref{extra-infinite-counter}) to $A^{e_0}$. Let
$A^{e_0}_c$ be the counter-zero automaton corresponding to
$A^{e_0}$. If
\begin{equation*}
{\cal{L}}({\pi}_{(\not= c_2)}{\lbrack}(A_{c_1}^{e_*} \cap_c
T(A^{e_*}_{c_2}))^{c_1>c_2}){\rbrack})=
{\cal{L}}(A_{c}^{e_*})\setminus {\cal{L}}(A_{c}^{e_0}),
\end{equation*} 
then $L(A^{e_*})\,\subseteq\,L(T^*(A))$.
\end{theorem}

\begin{proof}
By Proposition \ref{extra-finite} (respectively, Proposition
\ref{extra-infinite}), we have
$L(A^{e_*})=\bigcup_{i=0}^{\infty}L(A^{e_i})$.\\
\newline
\noindent
According to Lemma \ref{precis-extrapolation-auto2}, since
\begin{equation*}
{\cal{L}}({\pi}_{(\not= c_2)}{\lbrack}(A_{c_1}^{e_*} \cap_c
T(A^{e_*}_{c_2}))^{c_1>c_2}){\rbrack})=
{\cal{L}}(A_{c}^{e_*}),
\end{equation*} 
\noindent
we have
\begin{equation*}
 \forall w, {\forall} i >0 \ {\lbrack}w \in L(A^{e_i})\setminus
L(A^{e_0}) \ {\Rightarrow} \ {\exists}0\,{\leq}\,j<i, w \in
L(T(A^{e_j})){\rbrack}.
\end{equation*}

\noindent
It follows from Theorem \ref{precis-extrapolation-auto2} that
$L(A^{e_*})\,\subseteq\,L(T^*(A))$.
\end{proof}

\noindent
Theorem \ref{main-result-precis-auto} states a sufficient criterion to
check whether $A^{e_*}$ is a precise extrapolation of $T^*(A)$. This
criterion amounts to test whether Condition (\ref{cond-extra-auto1})
holds. For this, we proceed like for Condition
(\ref{cond-extra-transdu2}).\\
\newline

\noindent
Observe that, if $L(T^*(A))=L(A^{e_*})$, then the automata $A^{e_i}$
($i{\geq}0$) may constitute new elements in an extension of the
sampling sequence $S_I$, i.e., if $S_I=A^{s_0},A^{s_1},\dots,A^{s_k}$
with $A^{s_k}=A^{e_0}$, then the extension is
$A^{s_0},A^{s_1},\dots,A^{s_k},A^{s_{k+1}},$\\$A_{s_{k+2}},\dots$,
with $A^{s_{k+i}}=A^{e_i}$ for each $i{\geq}0$. Condition
(\ref{cond-extra-auto1}) is thus particularly designed to hold for
sampling sequences where each element can be obtained from the
previous one by a single application of the transducer $T$. Indeed,
the condition can be read as follows: {\em each automaton $A^{e_i}$ in
  the extended sampling sequence can be obtained by applying $T$ to an
  element that appears before in the sequence}. If more applications
of $T$ are needed, then we may have to adapt the condition. This is
illustrated with the following example.

\begin{example}
If each element in the sampling sequence is obtained by applying the
transducer $T$ $k>1$ times to the previous element in the sequence,
then one can test whether the following condition holds
\begin{equation}
\forall w, {\forall} i >0 \ {\lbrack}w \in L(A^{e_i})\setminus
L(A^{e_0}) \ {\Rightarrow} \ {\exists}0\,{\leq}\,j<i, w \in
L(T^k(A^{e_j})){\rbrack}.
\end{equation}
\noindent
rather than to check Condition (\ref{cond-extra-auto1}).
\end{example}

\noindent
This observation states for sampling sequences where the number of
applications of $T$ needed to build each element from the previous one
is constant. In \cite{Leg07}, we proposed another approach that
consists in associating to each state of the system an integer
variable that counts the number of applications of the reachability
relation needed to reach this state from the initial set of
states. Using this ``counter variable'', we can propose a preciseness
criterion whose induction is based on the number of applications of
the reachability relation rather than on the position in the sampling
sequence. Contrary to the techniques presented in this section, the
counters are no longer introduced during the extrapolation process,
but are present in all the steps of the computation. This is a ``key
point'' to ensure the preciseness when considering a nonlinear
sampling sequence, but this clearly influence the extrapolation
process and the increments detection. As observed in \cite{Leg07},
this approach is of particular interest when dealing with systems that
manipulate integer/real variables. However, the solution in
\cite{Leg07} is not a panacea. Indeed, as an example, it is known that
the transitive closure of the relation ${\lbrace}(x,2x){\rbrace}$ in
basis $2$ is regular, but the transitive closure of the relation
${\lbrace}((x,y),(2x,y+1)){\rbrace}$ is not regular.

\section{Implementation and Experiments}
\label{section-implementation}

This section briefly discusses an implementation of our results as
well as the experiments that have been conducted.

\subsection{Heuristics}
\label{heuristics}

Implementing the technique presented in this paper requires
potentially costly composition and determinization procedures. In
\cite{BLW03,BLW04a,Leg07}, we proposed two heuristics that, in some
situations, reduced to computation time from days to
seconds. Experimental results, which are presented in Chapter $7$ of
\cite{Leg07}, show that those heuristics are particularly useful when
working with arithmetic systems.

\subsection{The T(O)RMC Toolset}
\label{the-tool}

The results presented in this paper have been implemented in the {\em
  T(O)RMC} (states for {\em Tool for ($\omega$-)Regular Model
  Checking}) toolset\,\cite{Leg08}, which relies on the {\em LASH
  Toolset}\,\cite{LASH} for automata manipulations.

The LASH toolset is a tool for representing infinite sets and
exploring infinite state spaces. It is based on finite-state
representations, which rely on finite automata for representing and
manipulating infinite sets of values over various data domains. The
tool is composed of several C functions grouped into packages. The
LASH toolset implements several specific algorithms for solving the
($\omega$-)regular reachability problems of several classes of
infinite-state systems, which include FIFO-queue
systems\,\cite{BG96,BGWW97}, systems with integer
variables\,\cite{Boi03}, and linear hybrid
systems\,\cite{BHJ03,BH06}.\\
\newline
T(O)RMC extends the LASH toolset with the generic algorithm presented
in this paper. Contrary to the specific algorithms of LASH, the
algorithm of T(O)RMC is applicable to any system that can be
represented in the ($\omega$-)Regular Model Checking framework. This
makes it possible to handle classes of infinite-state systems that are
beyond the scope of specific algorithms, e.g., parametric
systems. T(O)RMC is divided into three packages, which are briefly
described hereafter.

\begin{enumerate}
\itemsep0cm
\item
{\em The transducer package} that provides data structures and
algorithms to manipulate transducers (composition, image computation,
$\dots$). The package also provides several heuristics to improve the
efficiency of the operations.
\item
{\em The extrapolation package} for detecting increments in a sequence
of automata, and extrapolating a finite sampling sequence. The tool
allows the user to precise (1) which sampling strategy has to be used,
and (2) how to build the successive elements in the infinite sequence.
\item
{\em The correctness package} that provides data structures and
algorithms to check the correctness of the extrapolation for several
classes of problems. The package also contains all the data structures
and algorithms to manipulate counter-word automata.
\end{enumerate}

T(O)RMC can be used to compute an extrapolation of a possibly infinite
sequence of automata $S=A^1, A^2,\dots$ . For this, the user has to
provide the following two functions:

\begin{itemize}
\item
A function named {\it SAMPLING} that takes as arguments two integers
$i$ and $j$. Each time T(O)RMC calls the function, it sets $i$ and $j$
to the indexes of two automata $A^i$ and $A^j$, such that $A^j$ is
incrementally larger than $A^i$. The function returns an automaton
$A^k$ which is assumed, by the user, to be the next automaton in a
sampling sequence whose two last elements are $A^i$ and $A^j$.
\item 
A function named {\it CHECK} that takes as argument an automaton
$A^{e_*}$. If the function returns yes, then T(O)RMC assumes that
$A^{e_*}$ is the extrapolation expected by the user. This is this
function that implements the checks for safety and preciseness.
\end{itemize}

\noindent
To extrapolate the infinite sequence of automata $S$,
T(O)RMC behaves as follows:

\begin{enumerate}
\item
T(O)RMC computes finite prefixes of $S$ until it finds two automata
$A^i$ and $A^{j}$ such that $A^{j}$ is incrementally larger than
$A^i$.
\item
T(O)RMC then tries to compute an incrementally growing sampling
sequence $S_I$, assuming that the two first elements of this sequence
are $A^i$ and $A^{j}$. The automata are added one by one to the
sampling sequence, using the function {\it{SAMPLING}}. Each time a new
automaton is added, the tool checks whether $S_I$ is still
incrementally growing. If no, then T(O)RMC goes back to point (1) and
consider a prefix of a longer size. If yes, then T(O)RMC extrapolates
$S_I$ and produces an automaton $A^{e_*}$. This extrapolation is
followed by a call to the function {\it CHECK} on $A^{e_*}$. If the
function returns yes, then the computation terminates, and $A^{e_*}$
is the automaton returned by the tool. If the function returns no,
then the tool tries to increase $S_I$ by adding one more automaton.
\end{enumerate}

\subsection{A brief Overview of the Experiments}

The T(O)RMC toolset has been applied to more than $100$ case
studies. This section only briefly recaps the classes of problems for
which T(O)RMC has been used so far. Details about the experiments
(including performances in terms of time and memory, which vary from
examples to examples) can be found in Chapters $7$ and $13$ of
\cite{Leg07}. 

We first used T(O)RMC to compute an automata-based representation of
the set of reachable states of several infinite-states systems,
including parametric systems, FIFO-queue systems, and systems
manipulating integer variables. Others experiments concerned the
computation of the transitive closure of several arithmetic
relations. It is worth mentioning that the disjunctive nature of some
relations sometimes prevents the direct use of specific domain-based
techniques~\cite{FL02,BH06}.  We also applied T(O)RMC to the
challenging problem of analyzing linear hybrid systems. One of the
case studies consisted of computing a precise representation of the
set of reachable states of several versions of the {\em leaking gas
  burner\/}. To the best of our knowledge, only the technique in
\cite{BH06} was able to handle the cases we considered. Among the
other experiments, we should also mention the computation of the set
of reachable states of an augmented version of the IEEE Root
Contention Protocol\,\cite{Leg07}, which has been point out to be a
hard problem\,\cite{SS01}.  The ability of T(O)RMC to compute the
limit of an infinite sequence of automata has other applications. As
an example, the tool has been used in a semi-algorithm to compute the
convex hull of a set of integer vectors \cite{CLW07,CLW08}. T(O)RMC
was also used to compute a symbolic representation of the simulation
relation between the states of several classes of infinite-state
systems\,\cite{BLW04b}.

The main goal of T(O)RMC is not performance improvement, but to allow
experimentation with automata sequence extrapolation in a variety of
context that goes beyond ($\omega$-)regular model checking
problems. As such T(O)RMC is slower than tools that are specific to
solving such model checking problems for the arithmetic domain
(e.g. FAST\,\cite{BLP06}, LIRA\,\cite{BDEK07}, LASH), but is perfectly
competitive when handling other regular model checking cases
(parametric systems, FIFO-queue systems, ...)\,\cite{RMC,VV06}.
T(O)RMC relies on LASH for automata manipulations. The LASH toolset is
oriented towards experimentation. It is thus less efficient for
manipulating automata representing sets of real/integer numbers than
LIRA and FAST that are oriented towards performances.


\section{A Brief Comparison with other Works}
\label{section-comparison}

In this section, we briefly compare our approach with other generic
techniques for solving the ($\omega$-)Regular Reachability Problems.\\
\newline
The Regular Model Checking framework has first been proposed in
\cite{KMMPS97} as a uniform paradigm for algorithmic verification of
parametric systems. The contributions in \cite{KMMPS97} are an
automata-based representation of parametric systems and an algorithm
to compute the transitive closure of the finite-word transducer
representing the reachability relation of such systems. One major
difference with our work is that the construction in \cite{KMMPS97}
can only be applied to a very specific class of finite-word
transducers.

In \cite{BJNT00,AJNd03}, Nilsson et al. proposed several
simulation-based techniques that, given a finite-word transducer $T$,
compute a finite-state representation for $T^+$. The core idea of
those techniques is to iterately compute the successive unions
$T^{{\leq}1},T^{{\leq}2},T^{{\leq}3},{\dots}$ (where
$T^{{\leq}i}={\bigcup_{n=1}^{i}}T^n$) and collapsing progressively
their states according to an equivalence relation, which is induced by
the simulation relations. The results of \cite{BJNT00,AJNd03} have
been implemented in a tool called the RMC toolset (states for Tool for
Regular Model Checking)\,\cite{RMC}, and tested on several parametric
and queue systems for which good results have been
obtained\,\cite{Nil05}. Unfortunately, it seems that the relations
used to merge the states of the successive unions have been designed
to handle parametric and queue systems only. To the best of our
knowledge, the RMC toolset cannot be used with other classes of
systems such as linear integer systems. In \cite{DLS02}, Dams,
Lakhnech, and Steffen proposed a non-implemented simulation-based
technique to compute $T^+$. This technique is similar to those
proposed in \cite{BJNT00,AJNd03}.

In \cite{Tou01,Tou03}, Touili proposed another extrapolation-based
technique to solve the Regular Reachability Problems. The results
presented in this paper share some notions with those in
\cite{Tou01,Tou03}. Indeed, the core idea in the work of Touili is to
compute an extrapolation of a finite-word transducer by comparing a
finite prefix of its successive powers, trying to detect increments
between them. One major drawback of Touili's work, which is not
implemented, is that no efficient method is provided to detect the
increments. There is no methodology to test whether the extrapolation
is precise or not. It is however easy to see that our preciseness
criterion directly adapts to Touili's extrapolation procedure.

In \cite{VSVA04,Vard06}, Vardhan et al. apply machine learning
techniques from \cite{Ang87,RS93} to learn a finite-word automaton
that represents the set of reachable states of a regular system. The
results in \cite{VSVA04,Vard06} have been implemented in a tool called
LEVER\,\cite{VV06}, which has been applied to FIFO-queue and linear
integer systems. A drawback with this approach is that it requires the
addition of witness variables that may break the regularity of the set
of reachable states. We also mention that in \cite{HV04}, Habermehl et
al. also proposed to use a learning-based approach to compute the set
of reachable states of several parametric systems.

Finally, even if they do not consider exactly the same problem as us,
it is relevant to mention a series of recent work~\cite{BHV04,BHMV05}
that combine abstraction-based techniques with automata-based
constructions to verifying reachability properties. Those works have
been shown to be particularly efficient for parametric and queue
systems\,\cite{BHV04} as well as for systems manipulating
pointers\,\cite{BHMV05}. On the other hand, one dedicated abstraction
is needed for each class of system, while our extrapolation-based
technique is designed to be applicable on any system that can be
represented by a ($\omega$-)regular system.




\section{Conclusion and Future Work}
\label{section-conclu}

In this paper, we have introduced an extrapolation-based technique for
solving the ($\omega$-)Regular Reachability Problems. The approach
consists in computing the limit of an infinite sequence of minimal
finite-word (respectively, minimal weak B\"uchi) automata by
extrapolating a finite sampled prefix of this sequence, i.e., selected
automata from a prefix of the sequence. The technique does not
guarantee that a result will be obtained, and correctness of the
guessed extrapolation needs to be checked once it is obtained. Our
results have been implemented in a tool called T(O)RMC, which has been
applied to several case studies.

One possible direction for future work would be to extend the
increment detection procedure described in Section
\ref{section-increment}. Indeed, as it is illustrated with the
following example, the procedure is not able to detect all possible
forms of increment.

\begin{figure}[t]
\centering 
\subfigure[\mbox{$A_1$}]{
\label{auto1}
\includegraphics[width=4cm]{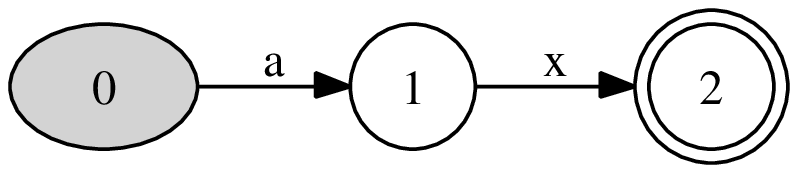}
}
\hspace{2em}
\subfigure[\mbox{$A_2$}]{
\label{auto2}
\includegraphics[width=4cm]{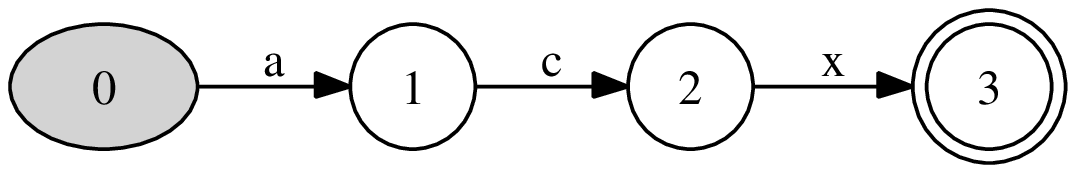}
}
\hspace{2em}
\subfigure[\mbox{$A_3$}]{
\label{auto3}
\includegraphics[width=4cm]{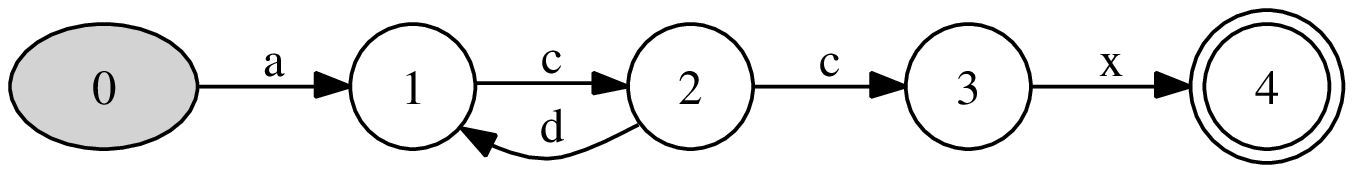}
}
\caption{Automata for Example \ref{example-observ2}.}
\label{observ2conclu-fig}
\end{figure}

\begin{example}
\label{example-observ2}
Consider the finite-word automata given in Figure
\ref{observ2conclu-fig}. The automaton $A_2$ differs from the
automaton $A_1$ by the addition of an increment, which is represented
by state $1$. If we compare $A_2$ and $A_3$, we see the addition of
one more increment. Clearly, $A_3$ differs from $A_1$ by the addition
of two increments represented by states $1$ and $2$. Unfortunately, in
$A_3$, the increment detected between $A_2$ and $A_3$ (state $2$ of
$A_3$) is the origin of a transition whose destination is the
increment detected between $A_1$ and $A_2$ (state $1$ of $A_3$). Such
a situation cannot be captured with the technique introduced in
Section \ref{section-increment}.
\end{example}

We could also investigate whether it is possible to detect the
repetition of different increment patterns in the same automaton. As
an example, the automata representing $ab$, $aabb$, $aaabbb$,
... differ by the repetitions of the symbols $a$ and $b$. If we
separately close those repetitions, we will obtain an automaton that
represents $a^+b^+$. This language, which is an over approximation of
the ``correct'' closure (i.e., $a^nb^n$ ($n\in \nats_0$)), may be
sufficient for practical applications. Another interesting direction
would be to extend our results to other classes of automata, which
includes tree and pushdown automata.

Another interesting direction would be to extend our results to other
classes of systems such as visibly pushdown systems\,\cite{AM04}. We
could isolate a class of systems for which we can always compute a
safe and precise extrapolation.

Finally, it would be of interest to extend ($\omega$-)Regular Model
Checking to the verification of {\em Open systems}. As opposed to
state-transition systems, open systems are systems whose behavior
depends on an external environment. In a series of fairly recent
papers, symbolic games\,\cite{ABd03,AHM01,BCFL05} have been proposed
as a general framework to specifying finite-state Open
systems\,\cite{AASFLRR06,AH01,ASFLRS05}. We believe that our work
could help to extending this approach to infinite-state open systems.






\section*{Thanks}

We thank Bernard Boigelot for a fruitful collaboration on preliminary
versions of this work. We also thank Marcus Nilsson, Julien d'Orso,
Parosh Abdulla, Sébastien Jodogne, Elad Shahar, Martin Steffen,
Tayssir Touili, and Mahesh Viswanathan for answering many questions
regarding their works and case studies.

\bibliographystyle{acmtrans}
\bibliography{paper}
\end{document}